\documentclass[runningheads,envcountsame]{llncs}

\newcommand{\longversion}[1]{#1}
\newcommand{\shortversion}[1]{}
\usepackage[utf8]{inputenc}

\usepackage[]{todonotes}
\usepackage{hyperref} 
\usepackage{mathtools} 
\usepackage{amsmath,amsfonts,amssymb} 

\usepackage[capitalize]{cleveref}
\usepackage{complexity}
\usepackage{xspace}

\usetikzlibrary{positioning, shapes,patterns,shadows.blur, arrows,decorations,arrows,automata,shadows,patterns,chains,graphs,calc,intersections,matrix,fit,shapes,chains,decorations.pathreplacing,chains,fit,shapes}

\DeclareMathOperator{\classical}{cl}
\DeclareSymbolFont{Shuffle}{U}{shuffle}{m}{n}
\DeclareFontFamily{U}{shuffle}{}
\DeclareFontShape{U}{shuffle}{m}{n}{%
  <-8>shuffle7%
  <8->shuffle10%
}{}
\DeclareMathSymbol\shuffle{\mathbin}{Shuffle}{"001}
\DeclareMathSymbol\cshuffle{\mathbin}{Shuffle}{"002}

\newcommand{\emptyword}{\lambda}

\def\ta{\mathtt{a}}
\def\tb{\mathtt{b}}
\def\tc{\mathtt{c}}
\def\td{\mathtt{d}}

\def\tn{\mathtt{n}}

\DeclareMathOperator{\alphabet}{alph}

\newcommand{\N}{\mathbb{N}}

\newcommand{\cG}{\mathcal{G}}

\newcommand{\cO}{\mathcal{O}}
\renewcommand{\cL}{\mathcal{L}}
\renewcommand{\cP}{\mathcal{P}}

\renewcommand{\alph}{\alphabet}

\newcommand{\iffl}{\shortversion{iff}\longversion{if and only if}\xspace}

\longversion{\usepackage[appendix=inline]{apxproof}}
\shortversion{\usepackage{apxproof}}%
\shortversion{}

\newif\ifinappendix
\newcommand{\appref}[1]{\shortversion{{$(*)$}}}
\newcommand{\applabel}[1]{
\shortversion{{$(*)$}}
\label{#1}}

\newtheoremrep{thm}[theorem]{Theorem}
\newtheoremrep{lem}[lemma]{Lemma}
\newtheoremrep{obs}[observation]{Observation}
\newtheoremrep{cor}[corollary]{Corollary}
\newtheoremrep{prp}[proposition]{Proposition}

\usepackage{comment}
\usepackage{lineno}
\nolinenumbers
\usepackage[T1]{fontenc}
\usepackage{graphicx}
\begin{document}
\title{On Languages Describing Large Graph Classes}
\author{Henning Fernau\inst{1}\orcidID{0000-0002-4444-3220} \and
Pamela Fleischmann\inst{2}\orcidID{0000-0002-1531-7970} \and 
Kevin Mann\inst{1}\orcidID{0000-0002-0880-2513} \and
Silas Cato Sacher\inst{1}\orcidID{0009-0004-6850-1298}
}
\authorrunning{Fernau et al.}
\institute{Trier University, Germany \email{\{fernau,mann,sacher\}@uni-trier.de} \and Kiel University, Germany \email{fpa@informatik.uni-kiel.de}}
\maketitle              % typeset the header of the contribution
\begin{abstract}
In this work, we introduce a new notion for representing graph classes with formal languages. In contrast to the seminal work by Kitaev and Pyatkin to represent graphs by words, we use formal binary languages in order to have a set of patterns (given by the languages' words) defining the edges in the graph. In particular, we investigate famous languages like the palindromes, copy-words, Lyndon words and Dyck words to represent all graphs or specific graph classes by restricting these languages.

\end{abstract}
\section{Introduction}
Graphs are a fundamental theoretical notion for describing relations between anything imaginable. Thus, the question how
to store efficiently formally given graphs in computers is as old as computer science itself, and the easiest way is the adjacency matrix which needs quadratic space in the number of the graph's nodes. In 2008 Kitaev and Pyatkin introduced the notion of word-representability \cite{KitPya2008} and even though it was meant to describe the Perkins semigroup, it also works as an elegant way to store a graph:
the graph's nodes are the word's letters and there is an edge between the nodes $u$ and $v$ iff the letters $u$ and $v$ alternate in the representing word, e.g., the {\em hour-glass graph} on five nodes $V=\{1,2,3,4,5\}$ where $1,2,3$ and $3,4,5$ build two triangles can be represented by $124534512$. Whereas every graph can be represented by an adjacency matrix, not all graphs are representable by words in this way (cf. the wheel graph $W_5$ \cite{KitPya2008,KitLoz2015}). Moreover, it is known to be \textsf{NP}-hard to decide whether a graph is representable or not. Such results motivate research into several directions including the one of altering the definition in order to represent all graphs.

Several variations and generalizations of word-representability have been described in the literature. The alternation of letters $\ta$ and $\tb$ in a word's projection onto $\{\ta,\tb\}$ equivalent to the avoidance of the factors $\ta\ta$ and $\tb\tb$. This leads to the generalization  via pattern avoiding words in~\cite{JonKPR2015}. For any pattern $u \in \{1,2\}^*$ such that $|u| \geq 2$, a graph $G=(V,E)$ is $u$-representable iff $\alphabet(w)=V$ and $\{\ta,\tb\} \in E$ iff $w$ avoids the pattern $u$ (cf. ~\cite[Def. 2]{JonKPR2015} for a formal definition). However, only in the case $u \in \{1\}^*$, their setting can be seen as defining abstract graphs and is hence comparable to classical word-representable graphs. Remarkably, for $|u| \geq 3$ every graph is $u$-representable. In particular, for $k \geq 3$, every abstract graph is $1^k$-representable. The {\em $k$-11-representability} generalizes this notion: a graph $G=(V,E)$ is \emph{$k$-11-represented} by a word $w\in V^*$ if, for all distinct vertices $\ta,\tb\in V$, at most $k$ times the patterns (factors) $\ta\ta$ and $\tb\tb$ occur in $w$'s projection onto $\{\ta,\tb\}$.  By definition, the $0$-11-representable graphs are the classical word-representable graphs, while it can be shown that the $2$-11-representable graphs are the class of all graphs  \cite{CheKKKP2019}. 
A different approach is the \emph{permutation representability} as defined in~\cite{KitSei2008}: here additional conditions on the words are given that can be used to represent graphs, not only on the patterns that define edges.
Lozin~\cite{Loz2008a} discussed several representations of graphs with finite automata. Another completely different way to connect the theory of regular languages with basic notions of graph theory, e.g., with bounded treewidth, was proposed in~\cite{DieFerWol2022}.  
In \cite{FleHLN2024a} it is shown that the language containing all words representing a specific graph is regular. However, there only one graph is investigated rather than an entire graph class.
Yet, there are other ways to connect words to graphs, and hence languages to graph classes. The notion of \emph{letter graphs} introduced by Petkovšek in~\cite{Pet2002}
gives rise to an infinite chain of graph classes where the words are simply restricted by the alphabet's size; in this context, this yields the graph parameter lettericity that recently obtained a certain popularity, cf.~\cite{AleAALZ2023} and related papers. A generalization of lettericity is contained in~\cite{FenFMRS2025}.

\noindent
\textbf{Own Contributions.}
Our approach seeks to generalize classical word-re\-pre\-sent\-ability by allowing for a broad variety of patterns in contrast to allowing or avoiding just one. As with classical word-rep\-re\-sent\-able graphs, the graph's nodes are the word's letters but there is an edge between the nodes $u$ and $v$ iff the pattern of $u$ and $v$ matches a word from a given binary language $L$, i.e., $h_{u,v}$ projects the word to $\{u,v\}$, maps $u$ to $0$, $v$ to $1$ and checks whether the resulting word is in~$L$. For example, $h_{u,v}(uuxvx) = 001$. The language $L_{\classical} \coloneqq(1\cup\emptyword)\cdot(01)^*\cdot (0\cup\emptyword)$ ($\emptyword$ is the empty word) is modeling classical word representability. Moreover, graph representation via pattern avoiding words, yields exactly the graphs characterized by $L_{\neg{u}}=\overline{\{0,1\}^*\{h_{1,2}(u)\}\{0,1\}^*}\cap \overline{\{0,1\}^*\{h_{2,1}(u)\}\{0,1\}^*}$ for a given pattern $u \in \{1,2\}^*$ with $|u| \geq 2$. Thus, language-representability can be understood as a generalization of graph representation via pattern-avoiding words. In particular, $L_{\neg{1^k}}$ describes the class of all graphs for each $k \geq 3$. The $k$-11-representable graphs can be modeled in our approach as well. 

In this work, we study the graph classes described by binary versions of famous languages such as the palindromes, the copy-language, the Dyck-language or Lyndon words. All of those, except for the Dyck-language, represent the class of all graphs and yield a representation that is as good as other implicit graph representations. The graphs described by the Dyck-language describe exactly the comparability graphs. We also show that we obtain specific graph classes if we intersect the aforementioned languages with other languages, e.g., the language for the classical word-representability.

\section{Preliminaries: Fixing General Notions and Notation}
\label{sec:prelims}

%\todohf{I still find this section simply horrible. There are so many notions introduced that are never used, in particular not in the first 12 pages. What a waste of space!}

In this section we introduce all notations and notions we need for our results and the framework in general.\longversion{Apart from the very brief next paragraph fixing notations for the integer numbers and set theory, this is a section containing one subsection for the field of {\em language theory (including combinatorics on words)} and one for the field of {\em graph theory}.}
Let $\N$ denote the natural numbers including $0$, let  $[n]=\{i\in\N \mid 1\leq i\leq n\}$ and  $\N_{\geq n}=\{k\in \N\mid k\geq n\}$ for some $n\in\N$.
%The cardinality of a set $X$ is denoted by $|X|$. 
For a set $X$ and $k\in\N$, define $\binom{X}{k}=\{Y\subseteq X\mid |Y|=k\}$. 
\paragraph*{Language Theory and Combinatorics on Words.}
\longversion{This subsection is dedicated to the definitions around the notion of {\em words} and sets of words, i.e. {\em languages}.
As in the natural understanding a word consists of letters (where the order matters) and a set of words is a {\em formal} language. Since our framework heavily relies on (different) formal languages in order to represent graphs, we need to introduce the notions from combinatorics on words to obtain the power of our framework.}
An \emph{alphabet} $\Sigma$ is a non-empty, finite set with \emph{letters} as elements. A \emph{word} over $\Sigma$ is a finite \emph{concatenation} of letters from~$\Sigma$; $\Sigma^{\ast}$ denotes the set of all words over~$\Sigma$, including the \emph{empty word}~$\emptyword$. The {\em length of a word $w$}, i.e., the number of its letters, is denoted by $|w|$. Let $\Sigma^+=\Sigma^{\ast}\setminus\{\emptyword\}$. 
%\longversion{ In other words, $\Sigma^*$ is the free monoid generated by~$\Sigma$, whose operation is called concatenation, mostly written by juxtaposition, only sometimes made explicit as~$\cdot$.} 
For $w\in\Sigma^{\ast}$ and for all $i\in[|w|]$, $w[i]$ denotes the $i^{\mbox{\tiny th}}$ letter of~$w$.
The number of occurrences of $\ta\in\Sigma$ in $w\in\Sigma^*$ is defined as $|w|_{\ta}=|\{i\in[|w|]\mid w[i]=\ta\}|$ and $w$'s alphabet is given by $\alphabet(w)=\{\ta\in\Sigma\mid\exists i\in[|w|]:\,w[i]=\ta\}$ as the symbols that occur in~$w$. Define $F(w)=\{|w|_a\mid a\in \Sigma\}$ as the 
{\em frequentnesses}\longversion{ of letters in~$w$}.
A word $w\in\Sigma^{\ast}$ is called \emph{$k$-uniform} for some $k\in\N$ if $|w|_{\ta}=k$ for all $\ta\in\Sigma$. 
\longversion{For instance $w=\mathtt{banana}$ is a word over the alphabet $\Sigma=\{\ta,\tb,\tn\}$ of length $6$, we have $|w|_{\ta}=3$, $w[6]=\ta$, and $w$ is not uniform for any $k\in\N$ since the letters do not have the same number of occurrences. The word $\mathtt{tamtam}$ is $2$-uniform.}
Define the \emph{reversal} $w^R$ of $w\in\Sigma^{\ast}$ by $w^R=w[|w|]\cdot w[|w|-1]\cdots w[1]$;  $w$ is a \emph{palindrome} if $w=w^R$ holds\shortversion{.}\longversion{, e.g., $\mathtt{step}\cdot\mathtt{on}\cdot\mathtt{no}\cdot\mathtt{pets}$.}
The word $w\in\Sigma^{\ast}$ is called a \emph{repetition} if there exist $u\in\Sigma^{\ast}$ and $k\in\N_{\geq 2}$ such that $w=u^k$ where $u^0=\emptyword$ and $u^k=uu^{k-1}$. Repetitions with $k=2$ are called \emph{copy-words} and words that are not a repetition are called {\em primitive}. \longversion{In addition to the concatenation of words, we need}\shortversion{We also use} the {\em shuffle product}
to combine two words:
for $u,v\in\Sigma^*$, define the {\em shuffle} by 
$u\shuffle v \coloneqq \{x_1y_1x_2y_2\cdots x_ny_n\mid \exists x_1,\dots,x_n, y_1,\dots,y_n\in\Sigma^*: u=x_1\cdots x_n\land v=y_1\cdots y_n\}$.
\longversion{As an example, consider the words $\mathtt{ml}$ and $\mathtt{eon}$; we can take the $\mathtt{m}$ from the first word, then the $\mathtt{eo}$ from the second, $\mathtt{l}$ from the first, and finally $\mathtt{n}$ from the second to get $\mathtt{meoln}$ as one word of the shuffle product. But also $\mathtt{melon}$, $\mathtt{emlon}$, and $\mathtt{mleon}$ are examples for words from the shuffle-product of these two words.}

One very famous set of words are the {\em Lyndon words}. \longversion{Their definition needs the notion of ordering words as e.g.  known from phone books. This requires an ordering
on the alphabet, like $\ta\prec\tb$, which is extended such that we have $\mathtt{apple}\prec\mathtt{banana}$. More formally, we extend}Given a linear ordering~$\prec$ on $\Sigma$, we extend it to the linear {\em lexicographical order} on~$\Sigma^*$, again denoted by $\prec$ and defined by $u\prec w$ \iffl either $u$ is a \emph{prefix} of~$w$, i.e., $w=ux$ for some $x\in\Sigma^{\ast}$, or $w=x\tb y_1$, $u=x\ta y_2$ for 
letters $\ta\prec \tb$ and $x,y_1,y_2\in\Sigma^{\ast}$. The second required basis for Lyndon words is the notion of {\em conjugates}: the words $uv$ and $vu$ (for $u,v\in\Sigma^*$) are known as \emph{conjugates}.
\longversion{For instance, $\mathtt{bellpepper}$ and $\mathtt{pepperbell}$ are conjugates. Intuitively, we obtain
the conjugates of a word by {\em rotating} it until we reach the original word again.} Thus,
$\Sigma^*$ is partitioned in\longversion{to} \emph{conjugacy classes}\longversion{ as being conjugate defines an equivalence relation on~$\Sigma^*$}. Finally,  a primitive word is a \emph{Lyndon word} if it is the lexicographically smallest element in its conjugacy class. If we want to make the ordering explicit, we also write $\prec$-Lyndon word, or $\prec$-smaller, etc. \longversion{For instance, the word $\mathtt{bellpepper}$ is the Lyndon word of its class since it is the smallest among all its conjugates (rotations).}
\longversion{The last part of the definitions from combinatorics on words introduces a specific mapping of words onto words.}
A mapping $f$ from the free monoid $\Sigma^{\ast}$ into another monoid is called a {\em morphism} if
$f(xy)=f(x)f(y)$ holds for all $x,y\in\Sigma^{\ast}$. Note that a morphism is uniquely defined by giving the images of all letters. For $A \subseteq \Sigma$,  the \emph{projective morphism} 
$h_A: \Sigma^* \rightarrow A^*$ is defined by  $h_A(\ta) = \ta$ for $\ta \in A$ and $h_A(\ta) = \emptyword$ otherwise. \longversion{Thus, we have $h_{\ta}(\mathtt{banana})=\ta\ta\ta$.} Let $\widetilde{\cdot}:\{0,1\}^*\rightarrow\{0,1\}^*$ be the \emph{complement morphism} mapping $0$ to $1$ and $1$ to $0$\longversion{. Here we have for instance $\widetilde{010}=101$. Note that the complement morphism is an involution, i.e.,}\shortversion{ which is an involution, i.e., } $\widetilde{{\widetilde{w}}}=w$.
\longversion{
\bigskip

We finish this section with the main definitions from formal language theory.}
Each \longversion{subset $L$ of $\Sigma^{\ast}$}\shortversion{$L\subseteq\Sigma^*$} is called a \emph{language} over~$\Sigma$\shortversion{.}\longversion{, e.g., the set of all Lyndon words is a language.} We extend the concatenation to languages, so that we can define powers of a language~$L$ and the \emph{Kleene star} of $L$ as $L^*=\bigcup_{n\in\N}L^n$ (analogously, we define other operations on words for entire languages). We define the \emph{symmetric hull} operator $\langle L\rangle\coloneqq L\cup \widetilde{L}$.   A language $L\subseteq \{0,1\}^{\ast}$ is called {\em $0$-$1$-symmetric} if $L=\widetilde{L}$ (i.e., $L=\langle L \rangle$). For instance, the language $L_{\classical} =(1\cup\emptyword)(01)^*(0\cup\emptyword)$ is $0$-$1$-symmetric. This also holds for its complement $\overline{L_{\classical}}$ where $\overline{L}$ denotes $\Sigma^{\ast}\backslash L$ for some language $L\subseteq\Sigma^{\ast}$. Let $\mathrm{freq}(L)=\{n\in\mathbb{N}_{\geq 1} \mid \exists w\in L: |w|_0 = n\}$ denote the set of frequentnesses in a $0$-$1$-symmetric language~$L\subseteq \{0,1\}^{\ast}$.
The language $\mathcal{D}=\{w\in \{0,1\}^* \mid \vert w\vert_0 =\vert w\vert_1 \wedge \forall i\in [\vert w\vert ]:\vert w[1..i]\vert_0 \leq \vert w[1..i]\vert_1 \}$ is the restricted Dyck-language with one pair of parentheses with $0$ for the opening and $1$ for the closing bracket (cf. $D_1^{\prime\ast}$ in \cite{Ber79}). More generally, $D_k^{\prime\ast}$ is the set of correctly parenthesized words formed with $k$  different types of parentheses, forming an alphabet $\Sigma$ of size~$2k$.
Identifying singleton sets with their elements, we can build \emph{regular expressions} from letters by using concatenation $\cdot$, union $\cup$, and Kleene star. We also include set complementation~$\bar\cdot$ and the \emph{shuffle}~$\shuffle$ when building such expressions. 

%\medskip

\paragraph*{Graph Theory.} 
Throughout this paper, we only consider undirected finite graphs. A graph~$G$ is a pair $(V,E)$ with the finite, non-empty set of vertices $V$ and the set of edges $E\subseteq\binom{V}{2}$. For a given graph $G$, let $V(G)$ and $E(G)$ denote its sets of vertices and edges respectively.
The cardinality of $V$ is called the \emph{order}
% order is used
of~$G$. If $\{u,v\}\in E$, $u$ is called a \emph{neighbor} of~$v$, and $N(v)\subseteq V$ are the neighbors of~$v$; $N[v]\coloneqq N(v)\cup\{v\}$ is the \emph{closed neighborhood} of~$v$. 
% neighbor and neighborhood are only used to define degree and universal and isolated vertices
The {\em degree} of \longversion{a vertex }$v\in V$ is\longversion{ defined as} $\deg(v)\coloneqq|N(v)|$. 
% degree is used
The \emph{complement} of the graph $G$, written $\overline{G}$, satisfies $\overline{G}=(V,\binom{V}{2}\setminus E)$. For a graph class $\mathcal{G}$, $\text{co-}\mathcal{G} \coloneqq \{ \overline{G} \mid G \in \mathcal{G}\}$. 
% \text{co-}\mathcal{G} is used
If $G_1=(V_1,E_1)$ and $G_2=(V_2,E_2)$ are graphs, then $\varphi:V_1\to V_2$ 
is a \emph{graph morphism} iff $\{u,v\}\in E_1$ implies $\{\varphi(u),\varphi(v)\}\in E_2$; a bijection $\varphi$ 
is a \emph{graph isomorphism} iff $\{u,v\}\in E_1\iff \{\varphi(u),\varphi(v)\}\in E_2$;\longversion{ if such an isomorphism exists,} then we write $G_1\simeq G_2$.
%graph isomorphism and \simeq are used
%A graph $G_1=(V_1,E_1)$ is a \emph{subgraph} of $G_2=(V_2,E_2)$ iff $V_1\subseteq V_2$ and $E_1\subseteq E_2$. $G_1$ is an  \emph{induced subgraph} of $G_2$ ($G_1=G_2[V_1]$) iff $E_1=E_2\cap \binom{V_1}{2}$.
%subgraph is not used
We call a graph with $n\in\N$ vertices and an empty edge set a {\em null graph}, denoted by $N_n$, and a graph with $n$ vertices and all possible edges a {\em complete graph}, denoted by $K_n$, i.e. $\overline{N_n}=K_n$.
%complete graphs are used
A graph with $n+m$ vertices, with $n,m\in\mathbb{N}_{\geq 1}$ is called a \emph{complete bipartite graph} $K_{n,m}=(V,E)$ if $V=U\cup W$, $U\cap W=\emptyset$, with $|U|=n$ and $|W|=m$ and $E$ contains all edges between $U$ and $W$ but no other edges. 
% complete bipartite graphs are used
A graph  with $n\in\N$ vertices is a \emph{path} $P_n=(V,E)$ if there is a linear ordering~$<$ on its vertex set such that $\{u,v\}\in E$ if either $u$ is the immediate predecessor of~$v$ in~$<$ or $u$ is the immediate successor of~$v$ in~$<$.
A graph  with $n\in\N$ vertices is a \emph{cycle} $C_n=(V,E)$ if it contains a path $P_n=(V,E')$ and one additional edge~$e$ that connects the two vertices of this path that have degree one.
A vertex in a graph is called 
%\emph{universal} if $N[v]=V$ and 
{\em isolated} if $N(v)=\emptyset$. 
%universal vertices are not used
%isoladed vertices are used
If $G_1=(V_1,E_1)$ and $G_2=(V_2,E_2)$ are two graphs with disjoint sets of vertices $V_1$ and $V_2$, then the (graph) \emph{union} of $G_1$ and $G_2$ is $G_1\cup G_2=(V_1\cup V_2,E_1\cup E_2)$.
% graph union is used 
%while the (graph) \emph{join}  of $G_1$ and $G_2$ is
%$G_1\nabla G_2=(V_1\cup V_2,E_1\cup E_2\cup \{\{x_1,x_2\}\mid x_1\in V_1, x_2\in V_2\})$.
%\longversion{We can describe many graph classes by our framework whose definitions are explained in}\shortversion{For the definition of many graph classes, we refer to} \url{graphclasses.org}.%\todohfInPlace{Here are the graph class definitions hidden.}
Let $G=(V,E)$ be a graph with a partition $V=V_1\cup V_2$. If $\binom{V_1}{2}\cap E=\emptyset$ and $\binom{V_2}{2}\cap E=\emptyset$, then $G$ is called \emph{bipartite}.  If $\binom{V_1}{2}\cap E=\emptyset$ and $\binom{V_2}{2}\cap E=\binom{V_2}{2}$, then $G$ is called a \emph{split graph}. If $\binom{V_1}{2}\cap E=\binom{V_1}{2}$ and $\binom{V_2}{2}\cap E=\binom{V_2}{2}$, then $G$ is called \emph{cobipartite}. $k$-colorability generalizes bipartiteness by partitioning the vertex set into $k$ independent sets instead of two.
If there exists a partial order on the vertex set~$V$ of  $G=(V,E)$ such that $u,v\in V$ are incomparable if and only if $\{u,v\}\in E$, then $G$ is a \emph{comparability graph}.

\paragraph*{Our Main Notions.}

After establishing the basics from language and graph theory, we are now able to introduce the main notions for our setting. We start with a definition, which allows us to compare words over an alphabet of a graph's nodes with a binary language. This new morphism $h_{\ta,\tb}$ can be viewed as the composition of the projection $h_{\{u,v\}}$ and the renaming isomorphism $\iota:\{u,v\}^*\to\{0,1\}^*$ with $u\mapsto 0$ and $v\mapsto 1$
for some given nodes $u,v$. For example, if we have $\ta,\tb,\tn\in V$ and $\mathtt{banana}\in V^{\ast}$, we get $h_{\ta,\tb}(\mathtt{banana})=1000$. If we have $0$-$1$-symmetric languages, it is unimportant whether $u$ or $v$ is mapped to $0$ or $1$ resp.; it only influences later the names of the represented nodes.

\begin{definition}
Let $V$ be an alphabet and fix $u,v\in V$ with $u\neq v$. Define $h_{u,v}:V^{\ast}\rightarrow\{0,1\}^{\ast}$ by $u\mapsto 0$, $v\mapsto 1$ and $x\mapsto \emptyword$ for $x\in V\setminus\{u,v\}$.
\end{definition}

\noindent
Based on $h_{u,v}$, we introduce the notion of $L$-representability of graphs.
\begin{definition}

Let $L\subseteq\{0,1\}^*$ be $0$-$1$-symmetric. Let $w\in V^*$ with $V=\alphabet(w)$.
We call a graph $G=(V,E)$ \emph{$L$-represented by $w\in V^{\ast}$} iff, for all $u,v\in V$ with $u \neq v$, we have
\(\{u,v\}\in E \Leftrightarrow h_{u,v}(w)\in L. \)
A graph  $G=(V,E)$  is \emph{$L$-representable} if it can be $L$-represented by some word $w\in V^*$. 
Let $G(L,w)$ denote the graph $G$ that is $L$-represented by $w\in\Sigma^{\ast}$ and let $\cG_L=\{G(L,w)\mid w\in\Sigma^*, \Sigma=\alphabet(w)\}$.
\end{definition}

As $0$-$1$-symmetry is essential for $G(L,w)$ and $\mathcal{G}_L$ being well-defined, we will assume this condition for any binary language used for defining graphs and graph classes in the following.
%We will also say that, e.g., $C_4$ is $L$-represented by some word $w\in V^*$ (with $|V|=4$), referring to some abstract graph.
We finish this section with some immediate insights and examples. Since $h_{u,v} \in L$ iff $h_{u,v} \notin \overline{L}$ for every language $L$ and vertices $u,v$, we have $G(L,w) = \overline{G(\overline{L},w)}$ and we get:

\begin{lemma}%[{\cite[Lemma 4.2.2]{FenFFMS2024}}]
\label{lem:compl}
For each $L \subseteq \{0,1\}^*, \mathcal{G}_{\overline{L}} = \text{co-}\mathcal{G}_{L}$.
\end{lemma}

\begin{lemrep}\applabel{lem:clique-representation} 
Let $V$ be an alphabet and $w\in V^*$. If $0^k\shuffle 1^\ell\subseteq L\subseteq\{0,1\}^*$ for all $k,\ell\in F(w)$, then $K_{|V|}$ is $L$-represented by~$w$.
\end{lemrep}

\begin{proof}
Consider $G=(V,E)=G(L,w)$. Let $u,v\in V$ be two arbitrary letters. By definition, $|w|_u,|w|_v\in F(w)$. Hence, $h_{u,v}(w)\in 0^{|w|_u}\shuffle 1^{|w|_v}\subseteq L$, i.e., $\{u,v\}\in E$. Therefore, $G$ is a complete graph.
\end{proof}

%Formally, we use the natural numbers as vertex names in order to fix a countable alphabet over which we form possible finite words and which hence serves as a potentially infinite source of vertex names. 
%For better readability, we will also use other vertex names. Before, we give insights and results in our framework, we describe some graphs with the help of some formal languages.

\longversion{As it can be quite tedious to list all words of a finite language, we defined for $L\subseteq\{0,1\}^*$, 
$\langle L\rangle\coloneqq L \cup \{ \widetilde{w} \mid w \in L \}\,.$
Clearly, $\langle \cdot\rangle$ is a hull operator, so that it is, in  particular, idempotent, i.e., $\langle L\rangle=\langle\langle L\rangle\rangle$. 
To further simplify notation and to help recognize patterns, for $w\in\{0,1\}^*$ let $\nu(w)$
denote the lexicographically smallest element of $\{w,\widetilde{w}\}$, 
called \emph{normal form} of~$w$. For example, $011$ is the normal form of $100$.
Let $\nu(L)=\{\nu(w)\mid w\in L\}$. Clearly, $\langle\nu(L)\rangle=\langle L\rangle$, so that we can take $\nu(L)$ as the normal form of~$L$. For finite languages, we also omit braces in this notation and list the normal forms in length-lexicographical order.
For example, $L=(0\shuffle 1)\shuffle\{0,1\}=\langle 001,010,011\rangle$ in the normal form presentation as $\nu(L)=\{001,010,011\}$. Without further explicit mentioning, this will be our convention for writing down finite binary languages.}
%\paragraph*{First Results}

\begin{lemrep}\applabel{lem:representability-reversal} 
A graph~$G$ can be $L$-represented \iffl it can be $L^R$-represented. 
\end{lemrep}

\begin{proof} 
For every word $w$ over an alphabet $V$, $G(L,w) = G(L^R,w^R)$. Hence, a graph~$G$ can be $L$-represented (by~$w$) \iffl it can be $L^R$-represented (by~$w^R$). 
\end{proof}
\longversion{By the previous proof, we obtain the following corollary.}

\begin{corollary}\label{cor:L-symmetry}
    If $L\subseteq \{0,1\}^*$ \longversion{be a language with}\shortversion{satisfies} $L=L^R$, then $G(L,w)= G(L,w^R)$ for each word $w$.
\end{corollary}

Intuitively, one may assume that  $G$ is $L$-representable \iffl $G$ is $(L\cup L^R)$-representable but consider given $L = \overline{\langle 0001\rangle }$ and $w = \ta\ta\ta\tb$, we have $G(L,w) = (\{\ta,\tb \}, \emptyset)$, $L \cup L^R = \{ 0, 1 \}^*$, and $\mathcal{G}_{L \cup L^R}$ is the class of complete graphs and $G \notin \mathcal{G}_{L \cup L^R}$.
We finish this section by briefly looking into $0$-$1$-symmetric languages~$L$ that are closed under reversal ($L=L^R$), e.g.,  $L_{\text{cl}}$ is closed under reversal. Note that $L_{\text{cl}}$ equals %$L_{\neg 1^k}$ 
\longversion{$$L_{\neg 1^k}=\overline{\{0,1\}^*\{0^k\}\{0,1\}^*}\cap \overline{\{0,1\}^*\{1^k\}\{0,1\}^*}$$}\shortversion{$L_{\neg 1^k}=\overline{\{0,1\}^*\{0^k\}\{0,1\}^*}\cap \overline{\{0,1\}^*\{1^k\}\{0,1\}^*}$}
for $k = 2$. By \cite{JonKPR2015}, we know that $\mathcal{G}_{L_{\neg 1^k}}$ corresponds to the family of graphs called $1^k$-representable. There it is also shown that, for every $k \geq 3$, \emph{every} graph is  $1^k$-representable., i.e.,  $\mathcal{G}_{L_{\neg 1^k}}$ is the class of all (undirected) graphs. For instance, we have for $k=3$, 
$G(L_{\neg 1^3},w) = ([4], \{ \{1,2\}, \{2,3\}, \{3,4\}, \{4,1 \},\{2,4\}\})$ with the word $w = 1124112341234$. More representations of the cycle $C_4$ can be found in \autoref{fig:languages-C4}.

\begin{figure}
%\centering
\begin{minipage}[m]{0.2\textwidth}
\begin{center}
 \begin{tikzpicture}
    \tikzset{every node/.style={fill = white,circle,minimum size=0.05cm}}
    
            \node[draw] (x1) at (0,0) {$1$};
            \node[draw] (x2) at (1,0) {$2$};
            \node[draw] (x3) at (1,-1) {$3$};
            \node[draw] (x4) at (0,-1) {$4$};
    
            \path (x1) edge[-] (x2);
            \path (x2) edge[-] (x3);
            \path (x3) edge[-] (x4);
            \path (x4) edge[-] (x1);
    \end{tikzpicture} 
%\caption{Graph classes represented by certain length-uniform and nearly length-uniform languages} \label{tab:graphClasses}
\end{center}
\end{minipage}
\begin{minipage}[t]{0.71\textwidth}
\begin{tabular}[width=.5\textwidth]{| l |  l |}
\hline 
$L$ & $w$ \\
\hline 
 $L_{\neg{1^3}}$  \scalebox{.68}{\cite[Thm. 4]{JonKPR2015}} & \tiny $2213212431124112341234$ \\

$L_{1112}$  \scalebox{.68}{\cite[Thm. 3.12]{GaeJi2020}} & \tiny $221324124311241312341234$ \\
$L_{2\text{-}11}$  \scalebox{.68}{\cite[Thm. 5.2]{CheKKKP2019}}  & \tiny $14323412413224314231243143124312$\\
\hline
$L_{\cP}$ \scalebox{.68}{(Thm.~\ref{thm:palindromes-get-all})} & \tiny $423121123142$  \\
$L_{\mathcal{C}}$ \scalebox{.68}{(Thm.~\ref{thm:copy-words-get-all})}  & \tiny $121324123142$ \\
$L_{\mathcal{L}}$ \scalebox{.68}{(Thm.~\ref{thm:Lyndon-words-get-all})}\negthinspace  & \tiny $111222333444123412341124113234234223224343433433444444$ \\
%\hline
$L_{\mathcal{D}}$ \scalebox{.68}{(Thm.~\ref{thm:dyck_comparability})}\negthinspace& \tiny $12341234132412341234$ \\
\hline
\end{tabular}

    \end{minipage}
    \caption{Language $L$ represents $C_4$ with the word $w$. }
    \label{fig:languages-C4}
\end{figure}
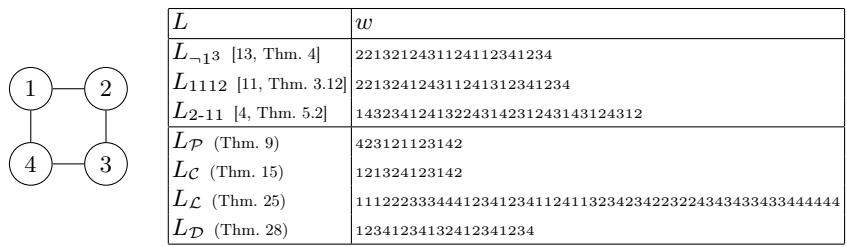

\section{Representing Graphs with Famous Languages}
In this section, we turn our attention to rather `famous' (infinite, non-regular) languages as the set of palindromes, the-copy language, the Dyck-language or the set of Lyndon words and look at their associated graph classes. Thus, in all paragraphs we start with the associated binary $0$-$1$-symmetric languages. We will see that most of these languages over the arbitrary alphabets allow to describe all graphs in a rather efficient way. Illustrations of these representations and the underlying constructions can be found in \autoref{fig:languages-C4}. Some proofs were omitted due to space restrictions; this is marked by $(*)$.

\medskip

\noindent
\textbf{Palindromes.}
Set $L_\mathcal{P}\coloneqq\{u\in\{0,1\}^+\mid u=u^R\}$. This language can be restricted to the deterministic context-free version of the binary palindrome language $L_{\text{detP}} :=\{ h_{\text{double}}(w) 01 (h_{\text{double}}(w))^R, h_{\text{double}}(w) 10 (h_{\text{double}}(w))^R \mid w \in \{0,1\}^{\ast}\}$ with the morphism $h_{\text{double}}: \{0,1\} \rightarrow \{0,1\}$ such that $h(\ta) = \ta\ta$ for all $\ta \in \{0,1\}$. Motivated by \cite{FleHLN2024a}, we also define
$\cG_{\star}=\{K_{1,n}\cup N_m\mid n,m\in\N\}$ as the set of all extended star graphs.
Regarding palindromes in general, notice that $h_{\ta,\tb}(w^R)=h_{\ta,\tb}(w)^R$ for each palindrome $w\in\Sigma^{\ast}$ and all $\ta,\tb\in\Sigma$. Our first lemma shows that being a palindrome can be characterized by the palindrome property of all projections to binary sub-alphabets. Besides being a nice word-combinatorial property, it\longversion{ also} has a graph-theoretic interpretation, see \autoref{cor:palindromes-and-complete-graphs}. \longversion{We will see similar properties for other well-known languages below.}

\begin{lemrep}\applabel{lem:binary-palindromes}
A word $w\in\Sigma^{\ast}$ is a palindrome \iffl, for all $\ta,\tb\in\Sigma$ with $\ta\neq\tb$, $h_{\ta,\tb}(w)$ is a palindrome.
\end{lemrep} 
\begin{proof}
W.l.o.g., we can assume $|\Sigma| > 1$ because for $|\Sigma| = 1$, the statement holds trivially. 
Let $w$ be a palindrome and $\ta,\tb\in\Sigma$. Choose $u\in\Sigma^{\ast}$ and $v\in\Sigma\cup\{\emptyword\}$ with $w=uvu^R$.
We have that $h_{\ta,\tb}(w)=h_{\ta,\tb}(u)h_{\ta,\tb}(v)h_{\ta,\tb}(u)^R$ which is a palindrome.\longversion{

}
Let now all projections onto two letters be palindromes. If $w[1]\neq w[|w|]$, then $h_{\{w[1],w[|w|]\}}(w)$ would not be a palindrome and thus we have $w[1]=w[|w|]$. Consider now $w[2..|w|-1]$. By the same argument, we get $w[2]=w[|w|-1]$. Inductively, we get that $w$ is a palindrome. \qed
\end{proof}

 The previous lemma links the set of all palindromes to  $L_\mathcal{P}$ and we show that the (deterministic) palindrome languages
 represents all graphs.

\begin{corollary}\label{cor:palindromes-and-complete-graphs}
A word~$w$ is a palindrome \iffl $G(L_\mathcal{P},w)\simeq K_{|\alph(w)|}$.
\end{corollary}

\begin{theorem}
\label{thm:palindromes-get-all}
  The binary palindrome language  $L_{\mathcal{P}}$  represent every graph. 
\end{theorem}
\begin{proof}
    Let $G=([n],E)$ be a graph with $n\in \N$. Define, for $i\in [n]$, $u_i\in [i-1]^*$ as the word which enumerates $ \overline{N(i)}\cap [i-1]$ in linear order as well as \shortversion{$w_1\coloneqq 11$ and $w_i\coloneqq iu_iw_{i-1}iu_i^R$ if $i\in\N_{\geq 2}$.}\longversion{    $$w_i\coloneqq \begin{cases}11,& i=1,\\
    iu_iw_{i-1}iu_i^R, & i\neq 1
    \end{cases}.$$} We want to show that $h_{j,k}(w_i) = h_{j,k}(w_i)^R$ \iffl $\{j, k\}\in E$ for all $i\in [n]$ and $j,k\in [i]$ with $j\neq k$. 
    For $i=1$ this trivially holds. So let $i>1$. Since $h_{j,k}$ is a morphism, 
    \shortversion{
    $
    h_{j,k}(iu_i w_{i-1} iu_i^R)^R 
    = h_{j,k}(u_i^R)^R h_{j,k}(i)^R h_{j,k}(w_{i-1})^R h_{j,k}(u_i^R) h_{j,k}(i)^R
    =h_{j,k}(u_i iw_{i-1}^R u_i^Ri)
    $.
    }
    \longversion{
    \begin{align*}
    h_{j,k}(iu_i w_{i-1} iu_i^R)^R 
    &= h_{j,k}(u_i^R)^R h_{j,k}(i)^R h_{j,k}(w_{i-1})^R h_{j,k}(u_i^R) h_{j,k}(i)^R\\ 
    &= h_{j,k}(u_i) h_{j,k}(i) h_{j,k}(w_{i-1})^R h_{j,k}(u_i^R) h_{j,k}(i)\\
    &=h_{j,k}(u_i iw_{i-1}^R u_i^Ri).
    \end{align*}
    }
    First we assume $i\notin \{j,k\}$. Then $h_{j,k}(w_i)= h_{j,k}(u_i) h_{j,k}(w_{i-1}) h_{j,k}(u_i^R)$. This equals $h_{j,k}(w_i^R) = h_{j,k}(u_i) h_{j,k}(w_{i-1})^R h_{j,k}(u_i^R)$ \iffl $h_{j,k}(w_{i-1})^R= h_{j,k}(w_{i-1})$. By induction, we know this is the case \iffl $\{j,k\}\in E$. 
    Now we can assume $j=i$. If $\{i,k\}\notin E$, then $h_{i,k}(w_i)= h_{i,k}(u_i)\cdot h_{i,k}(i)\cdot h_{i,k}(w_{i-1})\cdot h_{i,k}(u_i^R)\cdot h_{i,k}(i) = 1\cdot 0\cdot h_{i,k}(w_{i-1})\cdot 1\cdot 0 $ as $k<i$, which is no palindrome. For $\{i,k\}\in E$, $h_{i,k}(w_i)=0 h_{i,k}(w_{i-1})0$. This is a palindrome as $h_{i,k}(w_{i-1})\in 1^+$. By induction, $G\simeq G(L_\cP,w_n)$ follows. \qed
\end{proof}

\begin{theorem}
\label{thm:ldetp}
$\mathcal{G}_{L_{\text{\emph{detP}}}}$ is the class of all graphs. 
\end{theorem}

\begin{proof}
Let $G=([n],E)$ be a graph with $n\in \N$. Define $\langle M \rangle$ as the words that enumerate the vertices of the set $M \subseteq [n]$ in linear order. 
Define, for $i\in [n]$, $u_i := h_{\text{double}}(\langle \overline{N(i)}\cap [i-1] \rangle)$ and 
\longversion{
$$w_i\coloneqq 
\begin{cases}
11 \langle [n] \rangle 11,& i=1,\\
ii u_i w_{i-1} ii u_i^R, & i\neq 1
\end{cases}.$$
}
\shortversion{
$w_i\coloneqq 11 \langle [n] \rangle$ if $i=1$ and $w_i\coloneqq ii u_i w_{i-1} ii u_i^R$ if $i \neq 1$. 
}
We show $G(L_{\text{detP}},w_n)$ analogous to the proof of \autoref{thm:palindromes-get-all}. \qed
\end{proof}

By \cite[Thm. 28]{FleHLN2024a}, the non-empty palindromes represent exactly the extended star graphs in the following sense:
$\cG_{\star}=\{G(L_{\classical},w)\mid w\text{ is a palindrome}\,\}$. This raises the question if we can combine
the classical word-representability with our setting in the sense of intersecting $L_{\mathcal{P}}$ and $L_{\classical}$
to obtain the class of extended star graphs. The following theorem shows that in fact, we have $\cG_{\star}\subsetneq \cG_{L_{\mathcal{P}}\cap L_{\classical}}$, since the latter one is the class of bipartite graphs.

\begin{thmrep}\applabel{thm:alternating_palindromes}
%\shortversion{$(*)$}
    $\cG_{L_{\mathcal{P}}\cap L_{\classical}}$ is the class of bipartite graphs. 
\end{thmrep} 
\begin{proof}
Let $G=(V,E) \in \cG_{L_{\mathcal{P}}\cap L_{\classical}}$. Clearly, $L_{\mathcal{P}} \cap L_{\classical} = 1(01)^\ast \cup (01)^\ast 0 \cup \{ \emptyword \}$. Consider $\{ u, v \} \in E$. Obviously, $h_{\{u,v\}}(w) \in v(uv)^+ \cup (uv)^+ u$ because $h_{u,v}(w) \in L_{\mathcal{P}} \cap L_{\classical}$. If $h_{\{u,v\}}(w) \in v(uv)^+$, $|w|_v = |w|_u +1$ and if $h_{\{u,v\}}(w) \in u(vu)^+$, $|w|_u = |w|_v +1$. Hence, $E \subseteq \{ \{u, v \} \mid \exists \ell \in \mathbb{N}_{\geq 1}: |w|_v=\ell \wedge |w|_u=\ell+1\}$. Define $A\coloneqq\{v\in V\mid \exists \ell\in \mathbb{N}_{\geq 1}: |w|_v=2\ell\}$ and $B\coloneqq\{v\in V\mid \exists \ell\in \mathbb{N}_{\geq 1}: |w|_v=2\ell-1\}$. There is no edge between any two vertices of $A$ (or $B$), that is, $G$ is bipartite. 

    Let $G=(V,E)$ be a bipartite graph with the classes $A=\{a_1,\dots,a_s\}$ and $B=\{b_1,\dots,b_t\}$. For each $i\in [s]$, define $v_i\coloneqq a_1\cdots a_{i-1}a_{i+1}\cdots a_s$, $x_i$ as the word which enumerates $N(a_i)$ in linear order of indices of $B$, and $y_i$ as the word which enumerates $B\setminus N(a_i)$ in linear order of the indices of $B$, giving $u_i\coloneqq x_ia_iy_i$; define that $v_0=a_1\cdots a_s$, then,
    \longversion{$$w\coloneqq v_0u_1v_1\cdots u_sv_s\,.$$}\shortversion{$w\coloneqq v_0u_1v_1\cdots u_sv_s$.}
    
    Define $G'=(V,E')\coloneqq G(L_{\mathcal{P}}\cap L_{\classical},w)$. Clearly, for\longversion{ all} $j,k\in [s]$ with $j<k$, $h_{a_j,a_k}(u_iv_i)=01$ if $i\in [s]\setminus\{j,k\}$, $h_{a_j,a_k}(u_jv_j)=01$ and $h_{a_j,a_k}(u_kv_k)=10$, so totally, $h_{a_j,a_k}(w)=(01)^k10(01)^{s-k}\notin L_{\mathcal{P}}\cap L_{\classical}$, so there is no edge between vertices of $A$. For\longversion{ all} $j,k\in [t]$ with $j\neq k$, $h_{b_j,b_k}(w) \in (0\shuffle 1)^s$. Hence, $h_{b_j,b_k}(w) \notin L_{\mathcal{P}}\cap L_{\classical}$ because it has even length, so there is also no edge between vertices of~$B$.
    First consider $\{a_j,b_k\}\notin E$. Then $h_{a_j,b_k}(v_{j-1}u_j)=001$ is an infix of $h_{a_j,b_k}(w)$. Hence, $h_{a_j,b_k}(w)\notin L_{\mathcal{P}}\cap L_{\classical}$, i.e., $\{a_j,b_k\}\notin E'$.
    Secondly, consider $\{a_j,b_k\}\in E$. For $i\in [s]\setminus \{j\}$, $h_{a_j,b_k}(u_iv_i)=10$. Also $h_{a_j,b_k}(u_jv_j)=10$. Hence, $h_{a_j,b_k}(w)=0(10)^s\in L_{\mathcal{P}}\cap L_{\classical}$, so that $\{a_j,b_k\}\in E'$. Thus, $G=G'$. \qed
\end{proof}

\begin{comment} 
%This proof uses \autoref{lem:binary-palindromes}. We don't need it anymore.
\begin{proof}
 Let $G\in \cG_{\star}$.  By \cite{FleHLN2024}, there exists a palindrome~$w$ such that $G=G(L_{\classical},w)$. By \autoref{lem:binary-palindromes}, also $G=G(L_{\mathcal{P}}\cap L_{\classical},w)$. This proves the inclusion. To understand why this inclusion is proper, observe that $\cG_{\star}$ is hereditary but it is not closed under adding some twins (to the star center). Hence, properness follows with \autoref{thm:twins}. As a concrete example, take $w=repabper$ with $K_{2,3}=G(L_{\mathcal{P}}\cap L_{\classical},w)$, but clearly, this is not a star.
\end{proof}
\end{comment}

Since \autoref{thm:alternating_palindromes} shows that intersecting the palindrome language with the language for the classical word-representability may lead to interesting (even on first sight unexpected) results, we finish the paragraph about palindromes with two other combinations of languages.

\begin{correp}\applabel{cor:palindromes-yield-split-and-bipartite}
1.  $\mathcal{G}_{(L_\cP\cap L_{\classical})\cup ((0^2)^+\shuffle (1^2)^+)}$ is the class of split graphs. \\
2. $\mathcal{G}_{\overline{L_{\cP}}\cup \overline{L_{\classical}}}$ is the class of cobipartite graphs.
\end{correp}

\begin{proof}
By \autoref{thm:alternating_palindromes}, we get:\begin{enumerate}
\item $(L_\cP\cap L_{\classical})\cup ((0^2)^+\shuffle (1^2)^+)$ turns one partition side of a bipartite graph represented by $(L_\cP\cap L_{\classical})$ into a clique, so that we represent all split graphs this way.
    \item Using \autoref{lem:compl}, we see that $\overline{L_{\cP}}\cup \overline{L_{\classical}}$ represents all cobipartite graphs.   \qed
\end{enumerate}\renewcommand{\qed}{}
\end{proof}

\medskip

\noindent
\textbf{Copy-Language.}
Let $L_{\mathcal{C}}=\{ww \mid w\in \{0,1\}^{\ast}\}$ denote the binary copy-language which is $0$-$1$-symmetric. 
Apart from the fact that the language of palindromes is context-free, while the copy-language is not, both languages enjoy quite similar properties. We exemplify this by an analogue to \autoref{lem:binary-palindromes}. We follow a presentation similar to the one on palindromes, starting with a morphic characterization.
\begin{lemrep}\applabel{lem:binary-copies} 
For every alphabet $\Sigma$ such that $|\Sigma|\geq 2$, a word $w\in\Sigma^{\ast}$ is a copy-word \iffl, for all $\ta,\tb\in\Sigma$ with $\ta\neq\tb$, $h_{\ta,\tb}(w)$ is a copy-word.
\end{lemrep}
\begin{proof}
Consider a copy-word $ww$ over $\Sigma$. For each $\ta \in \Sigma$ and $\tb \in \Sigma \setminus \{ \ta \}$, $h_{\ta, \tb}(ww) = h_{\ta, \tb}(w) h_{\ta, \tb}(w)$ is a copy-word. Now consider a word $v \in \Sigma^\ast$ such that $h_{\ta, \tb}(v)$ is a copy-word for every $\ta \in \Sigma$ and $\tb \in \Sigma \setminus \{ \ta \}$. $|h_{\ta, \tb}(v)|_0$ and $|h_{\ta, \tb}(v)|_1$ are even for every $\ta \in \Sigma$ and $\tb \in \Sigma \setminus \{ \ta \}$. Hence, $|v|_{\ta}$ is even for every $\ta \in \Sigma$ and $v$ has even length. 
Let $\ta$ be the first letter of $v$. Hence, $u,u' \in \Sigma^\ast$ exist such that $v = \ta u \ta u'$ and $|u|_{\ta} = |u'|_{\ta}$. For every $\tb \in \Sigma \setminus \{ \ta \}$, $|u|_{\tb} = |u'|_{\tb}$ because, otherwise, $h_{\ta,\tb}(\ta u \ta u')$ would not be a copy-word. Hence $|u|_{\tc} = |u'|_{\tc}$ for every $\tc \in \Sigma$. Assume $v = x \td y x \td' y'$ for $\td, \td' \in \Sigma$ and $x, y, y' \in \Sigma^\ast$ such that $|y|_{\tc} = |y'|_{\tc}$ for every $\tc \in \Sigma$. Assume $\td \neq \td'$. 
$h_{\td, \td'}(v) = h_{\td, \td'}(x) \cdot 0 \cdot h_{\td, \td'}(y) \cdot h_{\td, \td'}(x) \cdot 1 \cdot h_{\td, \td'}(y')$ with $|y|_{\td} = |y'|_{\td}$ and $|y|_{\td'} = |y'|_{\td'}$. Since $h_{\td, \td'}(v)$ is a copy-word, $\td = \td'$. This is a contradiction. Hence, $\td = \td'$ and $v = x \td y x \td y'$ with $|y|_{\tc} = |y'|_{\tc}$ for every $\tc \in \Sigma$. By inductive reasoning, we get that $v$ is a copy-word. \qed
\end{proof}

Similarly to the palindromes, we can infer from \autoref{lem:binary-copies}
that we again obtain the complete graphs and in general that $\L_{\mathcal{C}}$ represents all words.

\begin{corollary}\label{cor:copy-language-clique}
A word~$w$ is a copy-word \iffl $G(L_\mathcal{C},w)\simeq K_{|\alph(w)|}$.
\end{corollary}

\begin{theorem}
\label{thm:copy-words-get-all}
The binary copy-language  $L_{\mathcal{C}}$  represent every graph. 
\end{theorem}

\begin{proof}
    Let $G=([n],E)$ be a graph. Define $w = u_1 1 u_2 2 \cdots u_n n 1 u_1 \cdots n u_n$ where $u_i$ enumerates each vertex of $\overline{N(i)}\cap [i]$ exactly once in the natural order for $i\in [n]$. Consider $G'=([n],E')=G(L_{\mathcal{C}},w)$. 
    Let $i,j\in [n]$ with $i<j$. By definition, $h_{i,j}(u_i)= \emptyword$. Also, $h_{i,j}(u_{j})\in \{\emptyword,0\}$ with $h_{i,j}(u_{j})= \emptyword$ being equivalent to $i\in N(j)$. For $k<i$, $h_{i,j}(u_k)= \emptyword$. 
    Hence, $h_{i,j}(w)= w_1w_2$ with 
    \begin{eqnarray*}w_1&\coloneqq&h_{i,j}(u_i)0h_{i,j}(u_{i+1})\cdots h_{i,j}(u_{j})1h_{i,j}(u_{j+1})\cdots h_{i,j}(u_{n})\text{ and}\\w_2 &\coloneqq& 0h_{i,j}(u_i)h_{i,j}(u_{i+1})\cdots 1 h_{i,j}(u_{j})h_{i,j}(u_{j+1})\cdots h_{i,j}(u_{n})\,.\end{eqnarray*}
    It is easy to see that $\vert w_1 \vert = \vert w_2 \vert$. Furthermore, $w_1=w_2$ holds \iffl $1 h_{i,j}(u_{j}) = h_{i,j}(u_{j})1$. Since $h_{i,j}(u_{j})\in \{\emptyword,0\}$, $1 h_{i,j}(u_{j}) = h_{i,j}(u_{j})1$ \iffl $h_{i,j}(u_{j})= \emptyword$. As $h_{i,j}(u_{j})= \emptyword$ \iffl $i\in N(j)$, we can conclude $E=E'$ and $G=G'$.\qed
\end{proof}

Estimating the complexity of a language by the Chomsky hierarchy,  the copy-language is quite complicated. Notice that by \autoref{lem:compl}, we get another characterization of all graphs through complementation, and the complement of the copy-language is a one-counter language, \longversion{i.e., it is in }a subclass of context-free.
\begin{corollary}
The complement of $L_{\mathcal{C}}$ can represent every graph. 
\end{corollary}

Now we intersect $L_\mathcal{C}$ with the the classical word-representation. Moreover,
we generalise this result: copy-words are only a special case of repetitions  $w^i$ for some $i\in\N_{\geq 2}$.

\begin{thmrep}\applabel{thm:copy-cap-classical}
$\mathcal{G}_{L_{\mathcal{C}}\cap L_{\classical}}$ is the set of classically word-representable graphs.
\end{thmrep}

\begin{proof} 
Note that $L_{\mathcal{C}}\cap L_{\classical} = \langle (0101)^* \rangle$. 

Let $G$ be (classical) word-representable. By \cite[Thm. 7]{KitPya2008}, $G$ is $k$-representable for some $k \geq 1$, i.e. there exists a $k$-uniform word $w$ such that $G = (L_{\classical}, w)$. Since every $k$-representable graph is also $(k+1)$-representable\cite[Obs. 4]{KitPya2008}, $G$ is $2\ell$-representable for some $\ell \geq 1$. Hence, there exists a $2\ell$-uniform word $w$ such that $G = (L_{\classical}, w)$. Note that $\langle(0101)^*\rangle \cap (0^{2\ell} \shuffle 1^{2\ell}) = \langle(0101)^\ell\rangle = L_{\classical} \cap (0^{2\ell} \shuffle 1^{2\ell})$. For all $\ta,\tb \in V$, $h_{\ta,\tb}(w) \in 0^{2\ell} \shuffle 1^{2\ell}$. Hence, $h_{\ta,\tb}(w) \in \langle(0101)^*\rangle$ iff $h_{\ta,\tb}(w) \in \langle(0101)^\ell\rangle$ iff $h_{\ta,\tb}(w) \in L_{\classical}$. Therefore, $G = (L_{\mathcal{C}} \cap L_{\classical}, w)$.

For the other direction, let $G \in \langle(0101)^*\rangle$. Hence, there exists a word $w$ such that $G = G(\langle(0101)^*\rangle, w)$. Let $V_{\text{even}} \coloneqq \{ \ta \in V | |w|_\ta \text{ is even} \}$. Note that all $\ta \in V \setminus V_{\text{even}}$ are isolated. Let $v \coloneqq h_{V_{\text{even}}}(w)$. We will prove that $G(L_{\classical}, v) = G(\langle(0101)^*\rangle, v)$. Since the word-representable graphs are closed under adding isolated vertices this shows that $G$ is word-representable. Let $\ta,\tb \in V_{\text{even}}$. There exists $\ell_1, \ell_2 \geq 1$ such that $\ell_1 = |w|_\ta$ and $\ell_2 = |w|_\tb$, i.e. $h_{\ta,\tb}(w) \in 0^{2\ell_1} \shuffle 1^{2\ell_2}$. If $\ell_1 = \ell_2$, we prove $h_{\ta,\tb}(w) \in \langle(0101)^*\rangle$ iff $h_{\ta,\tb}(w) \in L_{\classical}$ analogously to the other direction. If $\ell_1 \neq \ell_2$, $h_{\ta,\tb}(w) \notin \langle(0101)^*\rangle$ and $||w|_\ta - |w|_\tb|| \geq 2$. Therefore, $h_{\ta,\tb}(w) \notin L_{\classical}$. This shows  $G(L_{\classical}, v) = G(\langle(0101)^*\rangle, v)$. \qed
\end{proof}

\begin{lemrep}\applabel{lem:wtothei}
Let $w\in\Sigma^{\ast}$ and $i\in\N_{\geq 1}$. Then $w^i$ and $w$ $L_{\classical}$-represent the same graph \iffl for all $\ta,\tb\in\Sigma$ we have $h_{\{\ta,\tb\}}(w)\in \ta\{\ta,\tb\}^{\ast}\tb\cup\tb\{\ta,\tb\}^{\ast}\ta$. 
\end{lemrep}

\begin{proof}
Suppose for a contradiction that for some $\ta,\tb\in\Sigma$ we have that $h_{\{\ta,\tb\}}(w)$ starts and end with the same letter, w.l.o.g., with~$\ta$, and that $\ta,\tb$ alternate in~$w$. Then, they do not alternate in $h_{\{\ta,\tb\}}(w^2)$. On the other hand, if $\ta,\tb$ do not alternate in~$w$ (no matter whether or not they start and end in the same letter in $h_{\{\ta,\tb\}}(w)$), then they do not alternate in $w^2$, either. \qed
\end{proof}

\begin{proposition}\applabel{prop:rep}
Let $L=\langle 0\{0,1\}^{\ast}1\rangle$. %,1\{0,1\}^{\ast}0\}$ 
If $w$ $L$-represents~$G$, then also $w^i$ $L$-represents $G$ for all $i\in\N_{\geq 1}$.
\end{proposition}

We finish this paragraph with an -- in our opinion -- unexpected twist w.r.t. restricting the copy-language $L_{\mathcal{C}}$. Define for all $k\in\N$, the restriction $L_{\mathcal{C},k}\coloneqq\{ww\mid \vert w\vert_0\not \equiv_k \vert w\vert_1 \}$ of $L_{\mathcal{C}}$ where the number of $0$s and $1$ resp. is not allowed to be congruent modulo $k$. This language represents exactly the\longversion{ class of} $k$-colorable graphs. % and this result can be generalized towards \emph{coloring morphisms}. 

 \begin{thmrep}\applabel{thm:LCk}
     $\mathcal{G}_{L_{C,k}}$ is the class of $k$-colorable graphs.
 \end{thmrep}

\begin{proof}
 Let $G=([n],E)$ be a $k$-colorable with the coloring $f:[n]\to [k]$. Define $w = xu_1 1^k u_2 2^k \cdots u_n n^k x1^k u_1 \cdots n^k u_n$ where $u_i$ writes $t^k$ exactly once per $t\in \overline{N(i)}\cap [i]$ times in the natural order  for $i\in [n]$.  $x$ denotes $1^{f(1)}2^{f(2)}\cdots n^{f(n)}$.
    Consider $G'=([n],E')=G(L_{C,k},w)$. 
    Let $i,j\in [n]$ with $i<j$. By definition, $h_{i,j}(u_i)= \emptyword$. Also, $h_{i,j}(u_{j})\in \{\emptyword,0^k\}$ with $h_{i,j}(u_{j})= \emptyword$ being equivalent to $i\in N(j)$. For $k<i$, $h_{i,j}(u_k)= \emptyword$. Hence, $h_{i,j}(w)= w_1w_2$ with  
    \begin{eqnarray*}w_1&\coloneqq&0^{f(i)}1^{f(j)}h_{i,j}(u_i)0^kh_{i,j}(u_{i+1})\cdots h_{i,j}(u_{j})1^kh_{i,j}(u_{j+1})\cdots h_{i,j}(u_{n})\text{ and}\\
    w_2 &\coloneqq& 0^{f(i)}1^{f(j)}0^kh_{i,j}(u_i)h_{i,j}(u_{i+1})\cdots 1^k h_{i,j}(u_{j})h_{i,j}(u_{j+1})\cdots h_{i,j}(u_{n})\,.\end{eqnarray*}
    It is easy to see that $\vert w_1 \vert = \vert w_2 \vert$. Furthermore, $w_1=w_2$ holds \iffl $1^k h_{i,j}(u_{j}) = h_{i,j}(u_{j})1^k$ and $\vert w\vert_0 \not\equiv_k \vert w\vert_0$.
    
    If $f(i)=f(j)$ then 
    \begin{eqnarray*}
        \vert w_1\vert_0 &=& \vert   0^{f(i)}1^{f(j)}h_{i,j}(u_i)0^kh_{i,j}(u_{i+1})\cdots h_{i,j}(u_{j})1^kh_{i,j}(u_{j+1})\cdots h_{i,j}(u_{n})\vert_0\\
        &=& f(i) + \vert h_{i,j}(u_i)\vert_0+ k+ \vert  h_{i,j}(u_{i+1})\vert_0 +\cdots + \vert h_{i,j}(u_{j})\vert_0 + \vert h_{i,j}(u_{j+1})\vert_0\\
        &&+ \cdots  + \vert h_{i,j}(u_{n})\vert_0\\
        &\equiv_k& f(i)= f(j)\\
        &\equiv_k& f(j) +  \vert h_{i,j}(u_{j})\vert_1+ k + \vert h_{i,j}(u_{j+1})\vert_1 + \cdots  + \vert h_{i,j}(u_{n})\vert_1\\
        &=& \vert   0^{f(i)}1^{f(j)}h_{i,j}(u_i)0^kh_{i,j}(u_{i+1})\cdots h_{i,j}(u_{j})1^kh_{i,j}(u_{j+1})\cdots h_{i,j}(u_{n})\vert_1\\
        &=& \vert w_1 \vert_1.
    \end{eqnarray*}  
    Therefore, there is no edge and we can assume $f(i)\neq f(j)$.
    
    Since $h_{i,j}(u_{j})\in \{\emptyword,0^k\}$, $1 h_{i,j}(u_{j}) = h_{i,j}(u_{j})1^k$ \iffl $h_{i,j}(u_{j})= \emptyword$. As $h_{i,j}(u_{j})= \emptyword$ \iffl $i\in N(j)$, we can conclude $E=E'$ and $G=G'$.

    Conversely, let $G=([n],E)\in \mathcal{G}_{L_{C,k}}$. Then there exists  $w\in [n]^*$ with $G=G(L_{C,k},w)$. We can assume $2\mid \vert w\vert_i$ for all $i\in [n]$, as otherwise $h_{i,j}(w),h_{j,i}(w)$ cannot be in a copy language for all $j\in [n]$. Hence $N(i)=\emptyset$. This case can be bypassed by deleting $i$ from $w$ and adding $ii$ at the beginning of $w$. Then for all $j\in [n]$, $h_{i,j}(w)=001^t$ for some $t\in \mathbb{N}$.
    
    Define $f:[n]\to [k], i \to \frac{\vert w\vert_i}{2} \mod k$. Assume there are $i,j\in [n]$ such that $f(i)=f(j)$ but $\{i,j\}\in E$. Thus, there exists a $u\in V^*$ such that $h_{i,j}(w)=uu$. This implies $\vert u\vert_i = \frac{\vert w\vert_i }{2}\equiv_k \frac{\vert w\vert_j }{2} = \vert u\vert_j$. Therefore, $h_{i,j}(w)\notin L_{C,k}$.  \qed
\end{proof}

The construction of the words representing $G=([n],E)$ is similar to the idea of the copy-language: Let $f:V\to \{1,\ldots,k\}$ be the coloring. We use  $x=1^{f(1)}2^{f(2)}\cdots n^{f(n)}$ and $w = xu_1 1^k u_2 2^k \cdots u_n n^k x1^k u_1 \cdots n^k u_n$ where $u_i$ writes $t^k$ exactly once per $t\in \overline{N(i)}\cap [i]$ times in the natural order. Observe that $x$ ensures the coloring. For all $i\in [n]$, $u_i i^k$ and $ i^k u_i$ deletes the edges between $i$ and the elements in $\overline{N(i)}\cap [i]$. \shortversion{This construction can be generalized to graph morphisms.}

\begin{toappendix}
    
Let $G$ and $H$ be two (simple) graphs. A \emph{graph homomorphism} (or \emph{coloring morphism}) $f: V(G) \to V(H)$ from $G$ to the host graph $H$ satisfies the following condition:
If $\{u,v\}\in E(G)$, then $\{f(u),f(v)\}\in E(H)$. As $H$ is simple, all vertices of~$G$ mapped onto the same vertex of~$H$ form an independent set.
For example, if $H=K_k$, then a coloring morphism from $G$ to $K_k$ exists if and only if $G$ is properly $k$-colorable. 
Sometimes, one also allows non-simple host graphs~$H$.
To emphasize the host graph, such coloring morphisms are also called $H$-colorings.
\end{toappendix}

 \begin{thmrep}\applabel{thm:H-coloring}
    For all graphs $H=(U,F)$, there is a language $L(H)$ such that $\mathcal{G}_{L(H)}$ is the set of all graphs $G=(V,E)$ for which there exists a graph homomorphism $f:V\to U$.
 \end{thmrep}

\begin{proof}
    Let $H=(U,F)$ be a graph.
     Without loss of generality, we assume there is $k\in \mathbb{N}$ such  that $U=[k]$. 
    Define $L(H)\coloneqq\{ww\in \{0,1\}^*\mid \{\vert w\vert_0 \bmod{k}, \vert w\vert_1\bmod{k}\}\in F \}$ for all $k\in \mathbb{N}$.
    Let $G=([n],E)$ be a graph with the graph homomorphism $f:[n]\to [k]$. Define $w = xu_1 1^k u_2 2^k \cdots u_n n^k x1^k u_1 \cdots n^k u_n$ where $u_i$ writes $t^k$ exactly once per $t\in \overline{N(i)}\cap [i]$ times in the natural order  for $i\in [n]$.  $x$ denotes $1^{f(1)}2^{f(2)}\cdots n^{f(n)}$.
    Consider $G'=([n],E')=G(L_{C,k},w)$. 
    Let $i,j\in [n]$ with $i<j$. By definition, $h_{i,j}(u_i)= \emptyword$. Also, $h_{i,j}(u_{j})\in \{\emptyword,0^k\}$ with $h_{i,j}(u_{j})= \emptyword$ being equivalent to $i\in N(j)$. For $k<i$, $h_{i,j}(u_k)= \emptyword$. Hence, $h_{i,j}(w)= w_1w_2$ with  
    \begin{eqnarray*}w_1&\coloneqq&0^{f(i)}1^{f(j)}h_{i,j}(u_i)0^kh_{i,j}(u_{i+1})\cdots h_{i,j}(u_{j})1^kh_{i,j}(u_{j+1})\cdots h_{i,j}(u_{n})\text{ and}\\
    w_2 &\coloneqq& 0^{f(i)}1^{f(j)}0^kh_{i,j}(u_i)h_{i,j}(u_{i+1})\cdots 1^k h_{i,j}(u_{j})h_{i,j}(u_{j+1})\cdots h_{i,j}(u_{n})\,.\end{eqnarray*}
    It is easy to see that $\vert w_1 \vert = \vert w_2 \vert$. Furthermore, $w_1=w_2$ holds \iffl $1^k h_{i,j}(u_{j}) = h_{i,j}(u_{j})1^k$.

    Moreover,
    \begin{eqnarray*}
        \vert w_1\vert_0 &=& \vert   0^{f(i)}1^{f(j)}h_{i,j}(u_i)0^kh_{i,j}(u_{i+1})\cdots h_{i,j}(u_{j})1^kh_{i,j}(u_{j+1})\cdots h_{i,j}(u_{n})\vert_0\\
        &=& f(i) + \vert h_{i,j}(u_i)\vert_0+ k+ \vert  h_{i,j}(u_{i+1})\vert_0 +\cdots + \vert h_{i,j}(u_{j})\vert_0\\& +& \vert h_{i,j}(u_{j+1})\vert_0 + \cdots  + \vert h_{i,j}(u_{n})\vert_0\\
        &\equiv_k& f(i)\\
        \vert w_1 \vert_1 &=& \vert   0^{f(i)}1^{f(j)}h_{i,j}(u_i)0^kh_{i,j}(u_{i+1})\cdots h_{i,j}(u_{j})1^kh_{i,j}(u_{j+1})\cdots h_{i,j}(u_{n})\vert_1\\        
        &=& f(j) +  \vert h_{i,j}(u_{j})\vert_1+ k + \vert h_{i,j}(u_{j+1})\vert_1 + \cdots  + \vert h_{i,j}(u_{n})\vert_1 \\
        &\equiv_k& f(j)
    \end{eqnarray*}  
    Therefore, there is no edge and we can assume $\{f(i), f(j)\}\in F$.
    
    Since $h_{i,j}(u_{j})\in \{\emptyword,0^k\}$, $1 h_{i,j}(u_{j}) = h_{i,j}(u_{j})1^k$ \iffl $h_{i,j}(u_{j})= \emptyword$. As $h_{i,j}(u_{j})= \emptyword$ \iffl $i\in N(j)$, we can conclude $E=E'$ and $G=G'$.

    Conversely, let $G=([n],E)\in \mathcal{G}_{L(H)}$. Then there exists  $w\in [n]^*$ with $G=G(L(H),w)$. We can assume $2\mid \vert w\vert_i$ for all $i\in [n]$, as otherwise $h_{i,j}(w),h_{j,i}(w)$ cannot be in a copy language for all $j\in [n]$. Hence $N(i)=\emptyset$. This case can be bypassed by deleting $i$ from $w$ and adding $ii$ at the beginning of $w$. Then for all $j\in [n]$, $h_{i,j}(w)=001^t$ for some $t\in \mathbb{N}$.
    
    Define $f:[n]\to [k], i \to \frac{\vert w\vert_i}{2} \mod k$. Assume there are $i,j\in [n]$ such that $\{f(i),f(j)\}\notin $F but $\{i,j\}\in E$. Thus, there exists a $u\in V^*$ such that $h_{i,j}(w)=uu$. This implies $\vert u\vert_0 = \frac{\vert w\vert_0 }{2}\equiv_k f(i)$ and $\vert u\vert_1 = \frac{\vert w\vert_1 }{2}\equiv_k f(j)$. Therefore, $h_{i,j}(w)\notin L(H)$.  \qed
\end{proof}

\begin{toappendix}
    These results are interesting because of the following well-known dichotomy result for $H$-colorings due to Hell and Ne\v{s}et\v{r}il, see~\cite{HelNes90}, concerning the question of a $H$-coloring exists.
\begin{itemize}
    \item If $H$ is not bipartite but simple, then this question is \NP-complete.
    \item Otherwise, this question is polynomial-time solvable.
\end{itemize}
By the construction from \autoref{thm:H-coloring}, also carrying over the language $L(H)$, we can conclude:
\begin{corollary}
The recognition problem for $\mathcal{G}_{L(H)}$ is \NP-complete if and only if $H$ is not bipartite but simple.
\end{corollary}

\begin{proof}
The hardness carries over immediately. For membership, observe that the words that describe a graph of order~$n$ are of polynomial length in~$n$ by our construction. Hence, given a graph $G=(V,E)$, we can guess a word~$w$ that should describe this graph, plus a bijection between the set of vertices of $G(L(H),w)$ and $V$. Then, one can easily check if this bijection is a graph isomorphism.
\end{proof}
\end{toappendix}

\medskip

\noindent
\textbf{Lyndon Words.} 
\autoref{prop:rep} establishes that repetitions of a representant are not needed and even though
this word is not necessarily primitive, this insight encourages to look at the Lyndon words. By \cite[Prop.~5]{KitPya2008} we have that if a $k$-uniform word~$w$ represents a graph~$G$ in the classical sense, then every word in the conjugacy class of $w$ also represents~$G$, i.e., by \cite[Thm.~7]{KitPya2008}, for representing any word-representable graph, Lyndon words suffice. 
%Thus, we have a look whether there exists a class of graphs represented by Lyndon words within our framework. 
As we focus on $0$-$1$-symmetric languages and as not both $\ta\tb$ and $\tb\ta$ can be Lyndon for any two different $\ta,\tb\in\Sigma$, we look at the symmetric closure of binary Lyndon words: set $L_{\mathcal{L}}\coloneqq\langle L_{\text{Lyndon},\{0,1\},<}\rangle$ with the ordering $0<1$. Observe that $L_{\mathcal{L}}=L_{\text{Lyndon},\{0,1\},<}\cup L_{\text{Lyndon},\{0,1\}>}$. 
%Moreover, notice that t
The results for the copy-language imply that repetitions cannot occur in projections when dealing with Lyndon words. Therefore, we can, w.l.o.g., focus on the property of being the smallest conjugate in its class during the investigation of Lyndon words.
N\longversion{ote that n}ot all renaming projections $h_{\ta,\tb}(w)$ of a Lyndon word $w$ have to be Lyndon: we have for the Lyndon word $w=\ta\tb\tc\ta$ that \longversion{$h_{\{\ta,\tb\}}(w)=\ta\tb\ta\tb\ta$, i.e.,  }$h_{\ta,\tb}(w)=010$, is not Lyndon. %On the other hand, if all projections are Lyndon words, the word itself is Lyndon.
Hence, we only get\shortversion{:}\longversion{ a weaker projection lemma.}

\begin{lemrep}\applabel{projectionofLyndon}
Let $w\in\Sigma^{\ast}$ such that $h_{\ta,\tb}(w)$ is $<$-Lyndon for all $\ta, \tb\in\Sigma$ with $\ta  \prec \tb$. Then $w$ is $\prec$-Lyndon.
\end{lemrep}

\begin{proof}
Suppose that there exist $u,v\in\Sigma^+$ such that $w=uv$ and $vu\prec uv$. 
If $v$ were not a prefix of $u$, there existed $x\in\Sigma^{\ast}$ and some letters $\ta\prec \tb$ with $u=x\tb$ and $v=x\ta$. Thus, we have $w=x\tb x\ta$ and therefore $h_{\ta,\tb}(x)$ is not Lyndon. If $v$, on the other hand, were a prefix of~$u$, there existed $x_1\in\Sigma^+$ with $u=vx_1$. Thus, we have by our assumption $vvx_1\prec vx_1v$ and therefore $vx_1\prec x_1v$. Inductively, we obtain that $v=\ta$ and $u=\ta x_1\cdots x_k$ for some $k\in\N$. This leads to $w=\ta x_1\cdots x_k\ta$ and therefore $h_{\ta,\tb}(w)$ is not Lyndon for any other\longversion{ letter} $\tb\in\Sigma$. \qed
\end{proof}

\begin{lemrep}\applabel{lem:Lyndon_observation}
Let $\ta,\tb\in\Sigma$ with $\ta \neq \tb$ and $w \in \Sigma^{\ast}$. \\
1.  $h_{\ta,\tb}(w)$ is a $<$-Lyndon word %with respect to $0 \prec 1$ 
    \iffl $h_{\tb,\ta}(w)$ is a $>$-Lyndon word.\\ % with respect to $1 \prec 0$.
2. If $h_{\ta,\tb}(w)$ is a $<$-Lyndon word %with respect to $0 \prec 1$ 
    then $h_{\ta,\tb}(w)$ is no $>$-Lyndon word. % with respect to $1 \prec 0$.
\end{lemrep}

\begin{proof}
    \begin{enumerate}
        \item This statement is obvious as $h_{\ta,\tb}(w)[i]=0$ \iffl $h_{\tb,\ta}(w)[i]=1$ for $i\in [\vert h_{\tb,\ta}(w)\vert]\longversion{=[\vert h_{\ta,\tb}(w)\vert]}$.
        \item If $h_{\ta,\tb}(w)$ is a $<$-Lyndon word, there are $i,j\in [\vert h_{\ta,\tb}(w)\vert]$ such that $h_{\ta,\tb}(w)[i]=0$ and $h_{\ta,\tb}(w)[j]=1$. Especially, $h_{\ta,\tb}(w)[1]=0$. Otherwise, the conjugate starting at $i$ would be $<$-smaller than $h_{\ta,\tb}(w)$. This implies that the  conjugate starting at~$j$ is $>$-smaller than $h_{\ta,\tb}(w)$.\qed
    \end{enumerate}\renewcommand{\qed}{}
\end{proof}

\begin{theorem}
    \label{thm:Lyndon-words-get-all}
    $\cG_{L_\cL}$ includes every graph.
\end{theorem}
\begin{proof}
    Let $G=([n],E)$ be an graph with $n\in \mathbb{N}$. Define for $i\in [n]$ the words $v_i:=i\cdots n$, $x_i:=j_{i,1} \cdots j_{i,\ell_i}$, $y_i:=k_{i,1} \cdots k_{i,\ell'_i}$ and $u_i:=i^2 x_i i^2 y_i$ where $N(i) \cap [i..n] = \{j_{i,1},\dots, j_{i,\ell_i}\}$, $\overline{N(i)} \cap [i..n] = \{i,k_1,\dots, k_{i,\ell'_i}\}$ and $j_{i,1},\dots, j_{i,\ell_i}$ as well as $k_{i,1},\dots, k_{i,\ell'_i}$ are strictly monotone increasing. 
    We will show $G=G'=(V,E'):=G(L_{\mathcal{L}},w)$\longversion{ with $$w\coloneqq 1^3 2^3\cdots n^3 v_1^2 u_1 v_2^2 u_2\cdots v_n^2 u_n\,. $$}\shortversion{, $w\coloneqq 1^3 2^3\cdots n^3 v_1^2 u_1 v_2^2 u_2\cdots v_n^2 u_n$.} 

    Let $p,q,r\in [n]$ with $p<q$. By \autoref{lem:Lyndon_observation} only $h_{p,q}(w),h_{q,p}(w)$ with respect to\longversion{ the order} $0< 1$ needs to be considered. Since $p$ appears before $q$ in $w$, it is enough to consider $h_{p,q}(w)$. Observe that for $q<r$,  $h_{p,q}(v_r u_r) = \emptyword$. Hence, $h_{p,q}(v_{q+1}u_{q+1} \cdots v_{n}^2u_{n}) = \emptyword$. For $p<r<q$, $h_{p,q}(v_r^2u_r) = 111$ and $h_{p,q}(v_q^2u_q) = 111111$. Thus,  $h_{p,q}(v_{q}u_{q} \cdots v_{n}^2u_{n}) \in 111111(111)^*$.

    Next, assume $r<p<q$. Then $h_{p,q}(v_i)=01$. If $p,q\in N(r)$, $h_{p,q}(x_r)=01$ and $h_{p,q}(y_r)=\emptyword$. Analogously, $h_{p,q}(x_r)=\emptyword$ and $h_{p,q}(y_r)=01$ for $p,q\notin N(r)$. In the case of $p\in N(r)$ and $q \notin N(r)$, $h_{p,q}(x_r)=0$ and $h_{p,q}(y_r)=1$. For $p\notin N(r)$ and $q \in N(r)$, $h_{p,q}(x_r)=1$ and $h_{p,q}(y_r)=0$. Therefore, $h_{p,q}(v_r v_r u_r)\in \{010101,010110\}$ and $h_{p,q}(v_1^2u_1\cdots v_{p-1}^2 u_{p-1})\in \{010101,010110\}^{p-1}$. Clearly, there does not exist an $i \in [6(p-1)]$ with $w[3n+i..3n+i+2]=000$. 
    Finally, consider $r=p$. If $\{p,q\}\in E$ then $h_{p,q}(v_p^2,u_q)=010100100$ and $h_{p,q}(w)\in 000111\{010101,010110\}^{p-1}010100100111111(111)^*.$ 
    Clearly, $i=1$ is the only $i\in \vert h_{p,q}(w)\vert$ such that $h_{p,q}(w)[i..i+2]=000$. Since $h_{p,q}(w)[\vert h_{p,q}(w)\vert]=1$, $h_{p,q}(w)$ is Lyndon and $\{p,q\}\in E'$.

    Assume $\{p,q\}\notin E$. Then $h_{p,q}(v_pp^2 x_p p^2 y_p)=0100001$. As $v_p u_p= v_pp^2 x_p p^2 y_p$ is an infix of~$w$, $0100001$ is an infix of $h_{p,q}(w)$. As $000111$ is a prefix of $h_{p,q}(w)$, $h_{p,q}(w)$ is no Lyndon word (a conjugate with prefix $h_{p,q}(u_p)$ is $<$-smaller). Thus, $\{p,q\}\notin E$. 
    In total, \longversion{$E=E'$ and }$G=G'$. \qed
\end{proof}

We conclude this paragraph with some insights about graph classes when we only look at odd or even length Lyndon words
with a similar result to \autoref{cor:palindromes-yield-split-and-bipartite}.

\begin{thmrep}\applabel{thm:bipartite_Language}
$\mathcal{G}_{\{w\in L_\cL\mid |w| \text{ is odd}\}}$ is the class of bipartite graphs.
\end{thmrep}
\begin{proof}
    Assume $G=(V,E)$ is an $L$-representable graph. Then there exists a $w\in V^*$ such that $G=G(L,w)$. Define $A\coloneqq\{v\in V\mid \exists k\in \N :\: \vert w\vert_v =2k\}$ and $B\coloneqq\{v\in V\mid \exists k\in \N :\: \vert w\vert_v =2k + 1\}$. Clearly, there are no edges $\{a_1,a_2\}\in E \cap \binom{A}{2}$ as $h_{a_1,a_2}(w)$ is even. Analogously, there is no edge $\{b_1,b_2\}\in E \cap \binom{B}{2}$. Hence, $G$ is bipartite.

    Let $G=(V,E)$ be a bipartite graph with the classes $A=\{a_1,\dots,a_s\},B=\{b_1,\dots,b_t\}$. Let $y_i$ be the word that enumerates each element in $N(a_i)$ for $i\in [s]$ exactly once in order of the indices of the elements in $B$. Define \longversion{$$w\coloneqq a_1^3\cdots a_{s}^3 b_1^2 \cdots b_t^2 v_1\cdots v_s$$}\shortversion{$w\coloneqq a_1^3\cdots a_{s}^3 b_1^2 \cdots b_t^2 v_1\cdots v_s$} with $x_i \coloneqq a_i^2 y_i^2 a_i^2$ for $i\in [s]$. Clearly, for all $i\in [s]$ and $v\in V$, $\vert x_i\vert_v\in \{0,2,4\}$. Hence, $\vert w\vert_a$ is odd for $a\in A$ and $\vert w\vert_b$ is even for $b\in B$.
    
    Define $G'=(V,E'):=G(L,w)$.
    Since $L$ only includes words of odd length, there are no edges $\{a_1,a_2\}\in E' \cap \binom{A}{2}$ and $\{b_1,b_2\}\in E' \cap \binom{B}{2}$. 
    Hence, $G'$ is a bipartite graph with the classes $A,B$. Let $i\in [s]$ and $j \in [t]$.\longversion{

   } Assume $\{a_i,b_j\}\notin E$. Then $h_{a_i,b_j}(v_i)=0000$ is an infix and $h_{a_i,b_j}(a_1^3 \cdots a_s^3 b_1^2 \cdots b_t^2)=00011$ is a prefix of $h_{a_i,b_j}(w)$. Hence, $h_{a_i,b_j}(w)$ is no Lyndon word and $\{a_i,b_j\} \notin E'$.\longversion{

   } Assume $\{a_i,b_j\}\in E$. For $k \in [s] \setminus \{ i \}$, $\vert v_k\vert_{a_s} =0$ which implies $h_{a_i,b_j}(v_k)\in \{\emptyword,11\}$. Furthermore, $h_{a_i,b_j}(v_k) = 001100$. In total $h_{a_i,b_j}(w)\in 00011(11)^{\ast}001100(11)^{\ast}$. Therefore, $h_{a_i,b_j}(w)$ is Lyndon and $ \{a_i, b_j\}\in E$. Thus, $G=G'$. \qed
\end{proof}

\begin{corollary}
 Let $L\coloneqq \{w\in L_\cL\mid |w| \text{ is odd}\}$\longversion{ be the set of Lyndon words of odd length}, $L_{\text{odd}}:=\{w\in \{0,1\}^{\ast}\mid \vert w\vert_0,\vert w\vert_1 \text{ are odd}\}$ and $L_{\text{even}}:=\{w\in \{0,1\}^{\ast}\mid \vert w\vert_0,\vert w\vert_1 \text{ are even}\}$. Then  $\cG_{L \cup L_{\text{odd}}}$ is the class of split graphs and  $\cG_{L \cup L_{\text{odd}} \cup L_{\text{even}}}$ is  the class of  cobipartite graphs.
\end{corollary}

\begin{toappendix}
The classes of split and cobipartite graphs are uncomparable with the classical word-presentable graphs; this is studied in details in \cite{SriHar2025a,SriHar2025b}.
\end{toappendix}

\medskip

\noindent
\textbf{Dyck-language.}\label{subsec:Dyck}
Similar to the Lyndon words, we need here to define the symmetric closure of the Dyck-language as $L_{\mathcal{D}} := \langle 
\mathcal{D}\rangle$. Also similar, we have $h_{(,)}(w)\in L_{\mathcal{D}}$ for any pair of parentheses $(,)\in\Sigma$
if $w\in D_k^{\prime\ast}$, but the converse is not true, as $w=([)]$ shows. Our main result is that $\cG_{L_{\mathcal{D}}}$ is the class  of comparability graphs. To see that each graph in $\mathcal{G}_{L_\mathcal{D}}$ is a comparability graph, consider the order $\prec_V$ such that, for all $v,u\in V$, $v \prec_V u$ \iffl $v\neq u$, $\vert w\vert_v= \vert w\vert_u$ and, for all prefixes $w'$ of $w$, $\vert w'\vert_v\leq \vert w'\vert_u$. For the other direction, let $G=(V,E)$ be a comparability graph with the order $\prec$ on $V$. For each $v\in V$, $z_v\coloneqq y_v v x_v$ where $x_v$ enumerates all $u\in V$ with $ v\prec u$, and $y_v$ the remaining vertices. For all $v,u\in V$, if $ v\prec u$ then $h_{v,u}(z_v)=01$. The final word enumerates first all vertices and then all $z_v$ in the order of~$\prec$.

\begin{thm} \label{thm:dyck_comparability}
    $\cG_{L_{\mathcal{D}}}$ is the class  of comparability graphs. 
\end{thm}

%is  an enumeration of all vertices concatenated with an enumeration of all $z_v$ in the order of $\prec$.
\begin{proof}
    Let $G=(V,E)$ be a comparability graph w.r.t. the strict order $\prec$ on~$V$, i.e., $\{u,v\}\in E$ \iffl $u,v$ are comparable with respect to~$\prec$. Define the upper set $U_\prec(v)\coloneqq\{u\in V \mid v\prec u\}$ for $v\in V$. $\prec_{lin}$ denotes an arbitrary but fixed linear order on $V$ that extends~$\prec$. Define $z\coloneqq v_1\cdots v_n$ to be an enumeration of $V$ in the order $\prec_{lin}$\longversion{, i.e., $\vert V\vert = n$}.
    For $v\in V$, let $x_v,y_v$ be words such that $x_v$ enumerates $U_\prec(v)$ and $y_v$ enumerates $\overline{U_\prec(v)}\setminus \{v\}$ according to the order~$\prec_{lin}$ (increasingly). Further, denote $z_v\coloneqq y_v v x_v$, i.e., $z_v$ first enumerates the elements of~$V$ that are in the upper set of~$v$, then $v$, and then  the remaining elements of~$V$.
    Define $w\coloneqq z z_{v_1}\cdots  z_{v_n}$ and $G'=(V,E')=G(L_{\mathcal{D}},w)$. Note that for each $v,u\in V$, $\vert z_v\vert_u=\vert z \vert_u = 1$. Thus, for all $v,a,b\in V$, $h_{a,b}(z_v),h_{a,b}(z)\in \{01,10\}$. 
        Let w.l.o.g. $a,b\in V$ with $a\prec_{lin}b$ in our following discussion. Thus, $h_{a,b}(z) = 01$ $(*)$.
    
    Assume $\{a,b\} \in E$. As $G$ is a comparability graph, this implies $a\prec b$ $[*]$, as $b\not\prec a$  because $a\prec_{lin}b$. We show that $h_{a,b}(z_v) =01$ for all $v\in V$. Then, with $(*)$, we can conclude $h_{a,b}(w) =(01)^{n+1}$. Hence, $h_{a,b}(w)\in L_{\mathcal{D}} $. Suppose there existed $v \in V$ such that $h_{a,b}(z_v) =10$. Hence, $b$ precedes $a$ in the enumeration~$z_v$. If $v=a$ (or $v=b$, resp.), then $b\in \alph(y_v)=\overline{U_\prec(v)}\setminus \{v\}$ (or $a\in \alph(x_v)=U_\prec(v)$, resp.), contradicting $[*]$. Therefore, $v\notin \{a,b\}$. Since $x_v$, $y_v$ enumerate in accordance with~$\prec_{lin}$ and as $a\prec b$, $b$ has to appear in $y_v$ and $a$ in $x_v$. This is a contradiction to the transitivity of $\prec$, as $a\prec b$ and $v \prec a$ (because 
    %$\vert x_v \vert =1$
    $a\in U_\prec(v)$) would imply $b\in U_\prec(v)=\alph(x_v)$. Hence, $h_{a,b}(z_v) =01$ for all $v\in V$. 
    Now assume $\{a,b\} \notin E$. Then $\vert y_a\vert_b =1$ and $h_{a,b}(z_a) =10$. Let $a=v_i$. This implies $h_{a,b}(w) \in 01\{01,10\}^{i-1}10\{01,10\}^{n-i}$. Then $\{a,b\}\notin E'$, since $\vert h_{a,b}(w[1..2i+1])\vert_0 = i < i+1 = \vert h_{a,b}(w[1..2i+1])\vert_1$ and $\vert h_{a,b}(w[1])\vert_1 = 0 < 1 = \vert h_{a,b}(w[1])\vert_0$.

    Conversely, let $G=(V,E)\in \cG_{L_{\mathcal{D}}}$. There exists $w\in V^{\ast}$ such that $G=G(L_{\mathcal{D}},w)$. Define the order $\prec_V$ such that, for all $v,u\in V$, $v \prec_V u$ \iffl $v\neq u$, $\vert w\vert_v= \vert w\vert_u$ and, for all prefixes $w'$ of $w$, $\vert w'\vert_v\leq \vert w'\vert_u$. We show that $v \prec_V u$ is a strict partial order and that $\{v,u\}\in E$ \iffl $v\prec_V u$, i.e., $G$ is a comparability graph.
    By definition, $\prec_V$ is irreflexive. Let $a,b,c\in V$ with $a \prec_V b$ and $b \prec_V c$. Thus, $\vert w\vert_a= \vert w\vert_b = \vert w \vert_c $. For any prefix $w'$ of $w$, $\vert w \vert_a \leq \vert w \vert_b \leq \vert w \vert_c$. Let $i\in [\vert w\vert]$ be the the smallest~$i$ with $w[i]=b$. This implies $\vert w[1..i-1] \vert_a \leq 0 =\vert w[1..i-1] \vert_b$ and $\vert w[1..i] \vert_a \leq 0 < 1 =\vert w[1..i] \vert_b \leq \vert w[1..i]\vert_c$. Thus, $a\neq c$ and $a\prec_V c$. Therefore, $\prec_V$ is a strict partial order.

    Let $a,b\in V$. For $\vert w\vert_a\neq  \vert w\vert_b$, $\{a,b\} \notin E$ and $a\not\prec_V b$ as well as $b \not\prec_V a$. Thus, we can assume $\vert w\vert_a =  \vert w\vert_b$.  Since $h_{a,b}$ is a morphism, if $w'$ is a prefix of $w$, then $h_{a,b}(w')$ is a prefix of $h_{a,b}(w)$. Therefore, $\vert w'\vert_a =  \vert h_{a,b}(w')\vert_0$ and $\vert w'\vert_b =  \vert h_{a,b}(w')\vert_1$. 
    By definition,  $a \not \prec_V b$ and $b \not \prec_V a$ \iffl there are prefixes $w',w''$ such that $ \vert h_{a,b}(w')\vert_1 = \vert w'\vert_b < \vert w'\vert_a =  \vert h_{a,b}(w')\vert_0$ and $ \vert h_{a,b}(w'')\vert_0 = \vert w''\vert_a < \vert w''\vert_b  = \vert h_{a,b}(w'')\vert_1$.  Thus, $\{a,b\}\notin E$. 
    Conversely, assume $\{a,b\} \notin E$. Hence, there are prefixes $h',h''$ of $h_{a,b}(w)$  with $\vert h'\vert_1 < \vert h'\vert_0$ and $ \vert h''\vert_0 <  \vert h''\vert_1$. 
    Since $\vert h_{a,b}(v)\vert \leq 1$ for all $v\in V$, we can show by an inductive argument that, for each prefix $h$ of $h_{a,b}(w)$, there exists a prefix $w'$ of $w$ such that $h=h_{a,b}(w')$. Again by definition, this means that $a \not\prec_V b$ and $b \not \prec_V a$. \qed
\end{proof}

There is a rather unexpected connection between the number~$d$ of pairs of parentheses and the order dimension, a famous notion from order theory: Dyck words with~$d$ of pairs of parentheses describe exactly the comparability graphs of partial orders of dimension~$d$, as shown in~\cite{FenFFKS2026}. 

With \autoref{lem:compl}, we get a characterization of co-comparability graphs in terms of complements of the symmetric hull of the Dyck language. We finish this paragraph with the symmetric hull of probably the most famous Dyck language $L_1\coloneqq \langle\{0^n1^n\mid n\in\N\}\rangle$ and the non-regular language $L_2\coloneqq \{w\in\{0,1\}^*\mid |w|_0=|w|_1\}$ where
the Dyck-property is dropped (i.e., we have $2$-uniform words).
\begin{prprep}\applabel{propunioncointerval}
$\cG_{L_1}$ is the class of graph unions of co-interval graphs.\footnote{Co-interval graphs are not closed under graph union\longversion{, as, e.g.,}\shortversion{:} $2K_2$ is not a co-interval graph.}
$\cG_{L_2}$ is the class of cluster graphs.
\end{prprep}

\begin{proof}
Let $G=(V,E)$ be a co-interval graph. This means that we can associate to $V$ a family of intervals $\{I_v\mid v\in V\}$ such that $G$ can be viewed as the comparability graph of the corresponding interval order on~$V$, i.e., $u<v$ \iffl the rightmost endpoint of~$I_u$ is to the left of the leftmost endpoint of~$I_v$. Without loss of generality, we can assume that all endpoints are pairwisely distinct and all are integers from $[2|V|]$. Then, we can associate a word $w\in V^{2|V|}$ to $G$ by setting $w[i]=\ta$ (for $i\in [2|V|]$) \iffl one endpoint of $I_\ta$ is at position~$i$. Now, $G\simeq G(L_1,w)$ is easy to check. If we replace the first occurrence of every $\ta\in V$ by $\ta^k$ in~$w$, we arrive at a word $w_k$ of length $(k+1)|V|$. One can see that  $G\simeq G(L_1,w_k)$ for each positive integer~$k$. 
If we start with a union of $n$ co-interval graphs, yielding a graph~$H$, we can hence find a word $x_i$ of length $(i+1)|V_i|$ to describe component $H_i=(V_i,E_i)$. Define $x\coloneqq x_1\cdots x_n$.
As vertices occurring in $x_i$ do not occur in any other part of~$x$ and as only such vertices have frequentness~$(i+1)$, $H\simeq G(L_1,x)$ follows.

Conversely,  let $w\in \{0,1\}^*$.  Let $V=\alph(w)$. Consider $G_1=G(L_1,w)=(V,E_1)$.  Let $\ta,\tb\in V$ be two vertices and look at $h_{\ta,\tb}(w)$. This binary word belongs to $L_1$ if it is either $0^n1^n$ or $1^n0^n$ for some~$n$. In particular, this means that vertices of different frequentnesses in~$w$ are not connected. 
For each frequentness, however, we can build a separate interval model by defining the interval $I_\ta$ of $\ta$ by the position of the leftmost occurrence of $\ta$ in $w$ as the left endpoint and the position of the rightmost occurrence of $\ta$ as the right endpoint of~$I_\ta$. In this interval model, $\{\ta,\tb\}$ is an edge \iffl the intervals $I_\ta$ and $I_\tb$ are disjoint.
Hence, the whole graph $G_1$ is in fact the graph union of a number of co-interval graphs.

Let $G=(V,E)$ be a cluster graph, i.e., $G=\bigcup_{j=1}^k (V_k,E_k)$, where $G_k=(V_k,E_k)\simeq K_{n_k}$ for some positive integers $k$ and $n_1,\dots,n_k$. Accordingly, let $V_j=\{v_{j,1},\dots,v_{j,n_j}\}$ for $j\in [k]$. Define $w_j\coloneqq v_{j,1}^j\cdots v_{j,n_j}^j$ and $w\coloneqq w_1\cdots w_k$. 
Then, it is clear that $G=G(L_1,w)=G(L_2,w)$.

Conversely, let $w\in \{0,1\}^*$.  Let $V=\alph(w)$. Consider $G_2=(L_2,w)=(V,E_2)$. Let $\ta,\tb\in V$ be two vertices. If $|w|_\ta\neq|w|_\tb$, then $\{a,b\}\notin E_2$. Otherwise, $h_{\ta,\tb}(w)\in L_2$, so  $\{a,b\}\in E_2$. \qed
\end{proof}
By \autoref{lem:compl}, we get a characterization of complete multipartite graphs (as complements of cluster graphs) as $\cG_{\overline{L_2}}$.

\medskip

\noindent
\textbf{Size of the Graph Representations.}
The languages that we just discussed, most of which enable us to describe any graph, are not the first ones observed: for instance, $L_{\overline{1^3}}$ is such a language, as described in \cite{JonKPR2015}, see our discussions above.
However, there is one drawback of such an encoding: the word~$w(G)$ describing a given graph $G=(V,E)$ will have length $\varTheta(|V|!)$, so that we need an exponential number of bits to describe a graph, measured in its order. Even worse is the situation with 1112-avoidance, see \cite[Thm. 3.12]{GaeJi2020}, where the recursively constructed words can have length  $\cO(2^{|V|^2})$.
With 2-11-representable graphs, a more concise representation of every graph was proven in \cite[Thm. 5.1]{CheKKKP2019}, using only  $\cO(n^3\log(n))$ many bits. 
This is still worse than using traditional adjacency lists or adjacency matrices.
In particular, the latter data structure needs $\cO(n^2)$ bits to describe any (un)directed graph of order~$n$.
The good news is that the graph representations that we obtain with the help of palindromes, copy-words or Lyndon words are better. Taking as an example the cycle $C_4$, palindromes and copy languages performs much better, as illustrated in \autoref{fig:languages-C4}.

\begin{corollary}
With the help of palindromes, copy-words or Lyndon words, each graph of order~$n$
can be represented with $\cO(n^2\log(n))$ many bits. By \autoref{lem:compl}, similar statements are true for  complement languages.
\end{corollary}

Keep in mind that the $\log$-factor has to be added, as writing a single letter of the alphabet~$V$ needs $\cO(\log(|V|))$ many bits.
Notice that this matches other implicit graph representations, e.g., in terms of sum graphs, see~\cite{FerGaj2023}, or also that of adjacency lists\longversion{, a very prominent example of explicit graph representations}.
For sparse graphs, it is also important to take the number of edges into account.
For instance, for adjacency lists, one needs $\cO((n+m)\log n)$ many bits.
Next, we show that a slight modification of our constructions yields a graph representation that matches this bound.

\begin{thmrep}\applabel{thm:sparse-encoding}
There exists a language~$L$ such that any graph~$G=(V,E)$, with $|V|=n$ and $|E|=m$, can be represented as $G(L,w)$ for a word~$w\in V^*$ of length $\cO(n+m)$, so that $w$ can be encoded as a binary string of size $\cO((n+m)\log n)$.
\end{thmrep}

\begin{proof}
Consider $L=\overline{L_{\mathcal{C}}}$. Describe $G=(V,E)$ (of order~$n$) by the word~$w$ that is used to describe the complement of~$G$ according to the proof of \autoref{thm:copy-words-get-all}. By \autoref{lem:compl}, $G$ is isomorphic to $G(L,w)$. Now, notice that~$w$ contains $2n$ `vertex letters' and twice each of the words $u_i$ that can be described as enumerating all vertices of $N_G(i)\cap [i]$ (due to graph complement). Hence, the length of $u_i$ is bounded by the degree of vertex~$i$. By simple double counting, $|u_1\cdots u_n|\in\cO(m)$, where $m=|E|$, so that the claims follow immediately. \qed
\end{proof}

This theorem is also interesting from a language-theoretic perspective.
While the copy-language is not context-free, its complement is even a 1-counter language. Hence, there is no need to go very far in the Chomsky hierarchy to obtain such an efficient graph encoding. On the other hand, we neither have any arguments why regular languages have to lead to essentially bigger graph encodings. We leave this as an open question.

In contrast, it might be interesting to have encodings in hand that are suitable for storing dense graphs, as opposed to sparse graphs as in \autoref{thm:sparse-encoding}.
By the proof of this theorem, the language of copy-words $L_{\mathcal{C}}$ seems to be very suitable. For a concrete example, reconsider \autoref{cor:copy-language-clique}: for a vertex set~$V$ of size~$n$, $ww$ encodes $K_n$ if $w$ enumerates $V$ in some order.

\section{Conclusions and Further Open Questions}

In this paper, we have shown various examples of rather famous languages that can represent large classes of graphs, if not all graphs. For instance, we showed that both Lyndon words and palindromes are suitable to represent all graphs. This motivates the question what graph class Christoffel words (introduced in~\cite{Chr1875}) can represent. Notice that they are special Lyndon words and that for a Christoffel word $w=0v1$, $v$ is a palindrome. 
\begin{toappendix}
(More information on these words can be found in the monography~\cite{BerLRS2008}.)
\end{toappendix}
One of the nice combinatorial properties of classical word representability is that each word-representable graph can be represented by a uniform word. This kind of property does not hold for the language representations that we studied in this paper. However, there are indeed languages 
that can represent all graphs and that allow for uniform word representations. In~\cite{FenFFKS2026}, it is shown that every graph of so-called boxicity~$b$ can be represented by a finite $2b$-uniform language $L_b$, so that their union $L_B\coloneqq\bigcup_{b\geq 1}L_b$ represents all graphs as every graph has some finite boxicity~\cite{Rob69}. However, $L_B$ is no ``nice'' and well-known formal language, so that the question remains to find a ``known'' formal language representing all graphs by uniform words.%\todohfInPlace{I suggest removing the next sentence as we even KNOW that we cannot represent planar graphs.}
%We are also aware that the generalization of word representability introduced in~\cite{KenMal2023} to characterize all planar graphs cannot be modeled by our approach directly. Thus, further investigations are necessary.

\bibliographystyle{splncs04}
\bibliography{ab,hen,unpublished,wrg-addendum}

\appendix
%\todoscs{We need to remove the appendix!!}

\end{document}